\documentclass[a4paper,fleqn]{cas-dc}

\usepackage[authoryear]{natbib}

\usepackage{algorithm}
\usepackage{algorithmic}
\usepackage{enumitem}
\setlist[enumerate]{topsep=2pt,itemsep=0pt,parsep=0pt}
\usepackage{tikz}
\usetikzlibrary{arrows.meta, positioning, calc, fit, backgrounds, patterns}

\newdefinition{theorem}{Theorem}
\newdefinition{lemma}{Lemma}
\newdefinition{corollary}{Corollary}
\newdefinition{proposition}{Proposition}
\newdefinition{definition}{Definition}
\newdefinition{remark}{Remark}
\newdefinition{example}{Example}
\newdefinition{experiment}{Experiment}
\newdefinition{fact}{Fact}
\newproof{proof}{Proof}
\usepackage{etoolbox}
\AtEndEnvironment{proof}{\hbox{}\hfill$\mdwhtsquare$}
\newenvironment{proofof}[1]{\trivlist\item[\hskip\labelsep{\scshape Proof of #1.}]\rmfamily\ignorespaces}{\hbox{}\hfill$\mdwhtsquare$\endtrivlist}

\makeatletter
\def\@opargbegindefinition#1#2#3{\trivlist
	\item[\hskip\labelsep{\bfseries #1\ #2}]\upshape{\bfseries(#3).}\ \ignorespaces}
\def\@opargbegintheorem#1#2#3{\trivlist
	\let\baselinestretch\@blstr
	\item[\hskip\labelsep{\bfseries #1\ #2}]\itshape{\bfseries(#3).}\ \ignorespaces}
\makeatother

\newcommand{\BBR}{\mathbb{R}}
\newcommand{\BBN}{\mathbb{N}}
\newcommand{\BBZ}{\mathbb{Z}}
\newcommand{\BBD}{\mathbb{D}}
\newcommand{\SD}{\mathcal{D}}
\newcommand{\SC}{\mathcal{C}}
\newcommand{\SK}{\mathcal{K}}
\newcommand{\SU}{\mathcal{U}}
\newcommand{\exmark}{\hfill$\mdwhtdiamond$}
\newcommand{\exend}{\hfill$\mdwhtdiamond$}
\newcommand{\rmT}{\mathrm{T}}
\newcommand{\nn}{\nonumber}
\newcommand{\eps}{\varepsilon}
\newcommand{\matl}{\begin{bmatrix}}
	\newcommand{\matr}{\end{bmatrix}}
\DeclareMathOperator{\Real}{Re}
\DeclareMathOperator{\Dual}{Du}
\DeclareMathOperator{\diag}{diag}
\DeclareMathOperator{\relu}{ReLU}
\DeclareMathOperator{\col}{col}
\definecolor{mplblue}{HTML}{1F77B4}
\definecolor{mplred}{HTML}{FF0000}
\allowdisplaybreaks

\begin{document}
	\let\WriteBookmarks\relax
	\def\floatpagepagefraction{1}
	\def\textpagefraction{.001}
	\emergencystretch=2em
	\lineskiplimit=-2.5pt
	
	\shorttitle{A Dual Compiler for Neural Control Barrier Functions}
	\shortauthors{M.~Kamaldar}
	
	\title[mode = title]{Extracting exact Lie derivatives without backpropagation: A dual compiler for neural control barrier functions}
	
	\author[1]{Mohammadreza Kamaldar}[type=author,orcid=0000-0002-4774-9794]
	\cormark[1]
	\ead{mkamaldar@southalabama.edu}
	\credit{Conceptualization, Methodology, Software, Validation, Formal analysis, Investigation, Writing -- original draft, Writing -- review \& editing}
	
	\affiliation[1]{organization={Department of Mechanical, Aerospace, and Biomedical Engineering, University of South Alabama},
		addressline={150 Student Services Drive},
		city={Mobile},
		postcode={36688},
		state={AL},
		country={USA}}
	
	\cortext[1]{Corresponding author.}
	\begin{abstract}
		A safety filter based on a neural control barrier function (CBF) evaluates, at each control cycle, the trained network and its Lie derivatives along the system vector fields. This paper presents a compiler whose emitted C++ code performs this evaluation on embedded processors by forward-mode differentiation with dual numbers. The derivatives are exact in the sense of automatic differentiation, which means that they are free of truncation error and subject to floating-point roundoff. One evaluation of the network in this arithmetic, called a dual pass, returns the barrier value and its Lie derivative along one vector field, and an extension to hyper-dual numbers returns the second-order Lie derivatives of CBFs of relative degree two. Apart from scalar variables, the memory that a dual pass writes is at most four floating-point numbers per neuron of the widest layer and does not depend on the network depth. On an ESP32-S3 microcontroller, we compare the emitted code with code for reverse-mode differentiation (i.e., backpropagation), which is also generated before the program runs and uses fixed-size arrays. In three experiments, the reverse-mode code has a lower execution time than the dual and hyper-dual passes executed one after the other. For single-input networks, the measurements do not establish such an ordering between the reverse-mode code and a variant of the dual pass that propagates the real part once. For networks of increasing depth, the working memory of the dual passes is constant, whereas that of the reverse-mode code increases. While the emitted code runs, no dynamic allocation is recorded. We report the maximum observed execution time of the complete filter step, which includes the dynamics and a breakpoint search for the quadratic program with input bounds and one affine constraint. The compiler and the experiments are open source.
	\end{abstract}

	
	\maketitle
	
	\section{Introduction}
	\label{sec:intro}

	At each control cycle, a \emph{safety filter} based on a control barrier function (CBF) replaces a nominal control command by the closest command, in the Euclidean norm, that satisfies a safety constraint, and this constraint renders a prescribed set of states forward invariant \citep{ames2017cbf, ames2019cbf_qp}. If this set has no closed-form description, then the CBF can be parameterized by a feedforward neural network trained from data \citep{dawson2023safe, so2023ncbf}. The processor that runs such a filter evaluates, at each cycle, the trained network and its Lie derivatives along the system vector fields, because these derivatives are the coefficients of the safety constraint, which is affine in the input. A finite-difference approximation of these coefficients has a truncation error, and thus the constraint that the processor enforces differs from the constraint that the analysis assumes. This paper therefore considers exact differentiation, where \emph{exact} is meant in the sense of automatic differentiation (AD), namely, free of truncation error. The computed values remain subject to floating-point roundoff, overflow, underflow, and the errors of the library implementations of the elementary functions. In an embedded implementation, memory is a hard constraint \citep{david2021tflite}, and the memory that exact differentiation writes at each cycle thus restricts the processors on which the filter can run.

	Training frameworks compute exact derivatives by backpropagation, which is reverse-mode AD \citep{rumelhart1986backprop, abadi2016tensorflow, goodfellow2016deep, baydin2018ad_survey, paszke2019pytorch}. Reverse mode returns the gradient of a scalar-valued network from one forward sweep and one backward sweep, and the backward sweep uses quantities that the forward sweep computes at each layer. An implementation that computes these quantities once therefore stores them for all hidden layers, and its \emph{workspace}, which is the memory that it writes during one evaluation, grows with the sum of the hidden-layer widths. By default, PyTorch also records the operations of each call in a \emph{computational graph}, which it allocates dynamically on the \emph{heap}, namely, the memory that a program requests while it runs. Dynamic allocation complicates the analysis of the worst-case execution time (WCET) \citep{wilhelm2008worst} and is prohibited by Directive~4.12 of the MISRA~C coding standard for safety-critical embedded software \citep{misra2012guidelines, bagnara2018misra}. For a network of fixed topology, however, reverse-mode code can be generated before the program runs, that is, \emph{ahead of time}, with \emph{statically allocated} buffers, which are arrays of fixed size and location in memory. Section~\ref{sec:examples} uses such code as the principal baseline.

	This paper studies forward-mode differentiation with dual numbers \citep{clifford1871dual, wengert1964fast, baydin2018ad_survey} for the evaluation of the Lie derivatives of a neural CBF and presents an ahead-of-time compiler for this mode. A dual number $z=a+b\eps$ has a real part $a$ and a dual part $b$, where $\eps\neq0$ and $\eps^2=0$, and the extension of a differentiable function to dual arguments returns the function value as the real part and the directional derivative as the dual part. Applied layer by layer to a trained feedforward CBF network, this arithmetic computes a Lie derivative in one forward sweep, called a \emph{dual pass}, and the dual state of a layer is not used after the evaluation of the next layer. Apart from scalar variables, the workspace is thus bounded by four floating-point numbers, called \emph{floats}, per neuron of the widest layer and does not depend on the network depth. Forward mode, however, requires one sweep per vector field, whereas reverse mode obtains the Lie derivatives along all vector fields from one gradient (Example~\ref{ex:deep}). The compiler therefore emits two forms of this computation, namely, dual passes that are executed one after the other, which have the smaller workspace, and one sweep that propagates the real part once for all vector fields, which requires fewer operations and a larger workspace. On an ESP32-S3 microcontroller, we compare the emitted code with ahead-of-time, statically allocated reverse-mode code for the same networks. In three experiments, the reverse-mode code has the lower execution time, and, with increasing depth, the memory of the first form does not grow (Section~\ref{sec:examples}).

	CBFs build on the barrier certificates that \citet{prajna2004safety} use for the safety verification of hybrid systems. \citet{wieland2007cbf} use CBFs for constructive safety, \citet{xu2015cbf} study the robustness of CBFs, the quadratic program (QP) formulation used in this paper follows \citet{ames2017cbf}, and \citet{garg2024advances} present extensions that address input constraints, high relative degree, disturbances, and sampled-data effects. Applications on physical platforms include quadrotors among dynamic obstacles \citep{shayan2025ecbf} and a manipulator that operates near humans \citep{maithani2026proactive}. The learning of CBFs and other certificates from data \citep{robey2020learning, dawson2023safe, so2023ncbf} and the generation of barrier certificates by deep learning and symbolic regression \citep{ma2025datadriven} address the design of certificates and controllers and, to our knowledge, not the computation of the derivatives on an embedded target.

	Generated code and learned models are used in model predictive control \citep{kalmari2015toolkit, norouzi2022deep}, and frameworks that evaluate trained networks on microcontrollers, for example TensorFlow Lite Micro \citep{david2021tflite}, operate under the memory constraints of embedded systems and do not compute derivatives. Derivative code is generated ahead of time by the ACADO code-export module \citep{houska2011auto}, by CppADCodeGen, whose code runs in hard real-time loops with statically allocated memory \citep{giftthaler2017automatic}, and by CasADi, whose self-contained C code can be compiled on an embedded target \citep{andersson2019casadi}. Tapenade \citep{hascoet2013tapenade}, Enzyme \citep{moses2020enzyme}, and TinyAD \citep{schmidt2022tinyad} differentiate C and Fortran source code, programs in the LLVM intermediate representation, and C++ code, respectively. \citet{robenack2005automatic} computes Lie derivatives by AD for exact linearization. Second derivatives follow from families of univariate Taylor series \citep{griewank2000evaluating}, from hyper-dual numbers \citep{fike2011optimization}, and from dual-number formulations with three components \citep{penunuri2020dual} and of order up to four \citep{peon2024dual}. Dual algebras are used in kinematics \citep{bottema1990theoretical, montgomery2023lie} and in neural networks extended to dual numbers \citep{okawa2021dual, kozlov2022dualvalued}.

	The present work builds on established forward-mode differentiation with dual numbers, and its contribution is the application of this method to neural CBF filters on embedded processors, together with the analysis and the software for this application. General-purpose implementations of forward mode, for example \texttt{torch.func.jvp} and \texttt{jax.jvp}, perform the same differentiation within an AD framework, and the tools cited above differentiate general programs or expression graphs. The present compiler is instead specialized to the layer structure of a feedforward barrier network and to the Jacobian--vector-field products of the CBF constraint, for which the workspace and the operation counts admit closed-form expressions (Section~\ref{sec:algorithm}). In contrast to networks that are extended to dual numbers, the trained network is unchanged, and the compiler changes the arithmetic in which the network is evaluated. The certification of trained barriers lies outside the scope of the compiler.

	The contributions of this paper are as follows. First, we present an ahead-of-time compiler that translates barrier networks trained in PyTorch or stored in the Open Neural Network Exchange (ONNX) format into C++ \emph{header files} (i.e., source files that the user code includes) with the weights embedded as constant arrays, buffers of fixed size, no recursion, and no statements that allocate memory. Theorem~\ref{thm:main} states the exactness of the compiled computation with differentiability assumptions at the preactivations of the evaluated state, which is the form that applies to networks with the rectified linear unit (ReLU) activation. Second, we prove that, apart from scalar variables, the dual pass admits a workspace of at most four floats per neuron of the widest layer at all depths, whereas the workspace of a reverse-mode implementation that evaluates the forward sweep once grows with the sum of the hidden-layer widths (Proposition~\ref{prop:memory}). Third, we derive closed-form operation counts for the assembly of the CBF constraint by dual passes, by reverse mode, and by a \emph{shared-primal} sweep, which propagates the real part, called the primal in AD, once for all vector fields (Proposition~\ref{prop:dual_cost} and Remark~\ref{rem:ad_cost}). Fourth, we apply hyper-dual arithmetic \citep{fike2011optimization} to the second-order Lie derivatives of CBFs of relative degree two (Theorem~\ref{thm:hyper_lie}). Fifth, on an ESP32-S3 microcontroller, we compare the emitted code with ahead-of-time, statically allocated reverse-mode code for three nonlinear systems and for networks of increasing depth, and we report execution time, memory, code size, and numerical error (Section~\ref{sec:examples}).

	The compiler is released as version 0.2.0 of the open-source Python package \texttt{dual-cbf-compiler}. The test suite of the package asserts agreement between the generated C++ code and the \texttt{torch.autograd} reference to a relative and absolute tolerance of $5\times10^{-5}$ in single precision for the first-order quantities and of $10^{-3}$ for the second-order quantities. The Data availability statement lists the repositories of the compiler and of the embedded experiments of Section~\ref{sec:examples}.
	
	\section{Notation}
	\label{sec:notation}
	Let $\BBN\triangleq\{0,1,2,\ldots\}$, let $\BBZ_+\triangleq\{1,2,3,\ldots\}$, and let $\BBR$ denote the set of real numbers. Let $\BBR^n$ denote the $n$-dimensional Euclidean space and $\BBR^{n\times m}$ the set of $n\times m$ real matrices, and we identify $\BBR^1$ and $\BBR^{1\times1}$ with $\BBR$. For all $x\in\BBR^n$ and all $i\in\{1,\ldots,n\}$, let $x_{(i)}$ denote the $i$th component of $x$, let $\|x\|\triangleq(x^\rmT x)^{1/2}$ denote the Euclidean norm of $x$, and, for a vector with a subscript, such as $v_k\in\BBR^n$, let $v_{k(i)}\triangleq(v_k)_{(i)}$. For all $A\in\BBR^{n\times m}$ and all $j\in\{1,\ldots,m\}$, let $\col_j(A)$ denote the $j$th column of $A$. Let $\mathbf{1}_m\in\BBR^m$ denote the vector of ones, $I_n$ the $n\times n$ identity matrix, $\mathrm{vec}(\cdot)$ the column-stacking operator, and, for all $x\in\BBR^n$, $\diag(x)\in\BBR^{n\times n}$ the diagonal matrix whose $i$th diagonal entry is $x_{(i)}$. 
	
	Let $\SD\subseteq\BBR^n$ be open, let $\varphi\colon\SD\to\BBR^p$, and let $x\in\SD$. If $\varphi$ is differentiable at $x$, then $\varphi'(x)\triangleq\partial\varphi(x)/\partial x\in\BBR^{p\times n}$ denotes the Jacobian of $\varphi$ at $x$, and, if $p=1$, the column vector $\varphi'(x)^\rmT\in\BBR^n$ is the gradient of $\varphi$ at $x$. If $p=1$ and $\varphi$ is twice differentiable at $x$, then $\varphi''(x)\in\BBR^{n\times n}$ denotes the Hessian of $\varphi$ at $x$, which is the Jacobian of the gradient of $\varphi$ at $x$. Let $h\colon\SD\to\BBR$, $f,\xi\colon\SD\to\BBR^n$, and $G\colon\SD\to\BBR^{n\times m}$. If $h$ is differentiable at $x$, then the \emph{Lie derivatives} of $h$ along $\xi$ and $G$ at $x$ are $L_\xi h(x)\triangleq h'(x)\xi(x)\in\BBR$ and $L_Gh(x)\triangleq h'(x)G(x)\in\BBR^{1\times m}$, and thus, for all $j\in\{1,\ldots,m\}$, $\col_j(L_Gh(x))=L_{\col_j(G)}h(x)$, where $\col_j(G)$ denotes the vector field $x\mapsto\col_j(G(x))$. If $h$ is differentiable on $\SD$ and $L_fh$ is differentiable at $x$, then the \emph{second-order Lie derivatives} of $h$ at $x$ are the Lie derivatives of the function $L_fh\colon\SD\to\BBR$ at $x$, namely, $L_\xi L_fh(x)\in\BBR$, $L_f^2h(x)\triangleq L_fL_fh(x)$, and $L_GL_fh(x)\in\BBR^{1\times m}$. If $\varphi$ is twice differentiable at $x$, then, for all $v,w\in\BBR^n$, $\varphi''(x)[v,w]\in\BBR^p$ denotes the bilinear action of the second derivative of $\varphi$ at $x$ on the pair $(v,w)$, which is defined, for all $i\in\{1,\ldots,p\}$, by $(\varphi''(x)[v,w])_{(i)}\triangleq v^\rmT\varphi_{(i)}''(x)w$, where $\varphi_{(i)}\colon\SD\to\BBR$ denotes the $i$th component of $\varphi$. If $p=1$, then $\varphi''(x)[v,w]=v^\rmT\varphi''(x)w$. A continuous function $\alpha\colon\BBR\to\BBR$ is an \emph{extended class $\SK$ function} if it is strictly increasing and $\alpha(0)=0$.
	
	\section{Problem Formulation}
	\label{sec:problem}
	
	Consider the control-affine system
	\begin{equation}
		\dot{x}(t)=f\bigl(x(t)\bigr)+G\bigl(x(t)\bigr)u(t),
		\label{eq:dynamics}
	\end{equation}
	where, for all $t\ge0$, $x(t)\in\SD\subseteq\BBR^n$ is the state and $u(t)\in\SU\subseteq\BBR^m$ is the control input, $\SD$ is open, $\SU$ is nonempty, compact, and convex, $f\colon\SD\to\BBR^n$ is the drift vector field, and $G\colon\SD\to\BBR^{n\times m}$ is the input matrix. We assume that the state $x$ is available for feedback, and we consider the following assumption.
	\begin{enumerate}[label={A\arabic*)},ref={A\arabic*)},leftmargin=*,labelindent=3pt,topsep=4pt]
		\item\label{assum:dynamics} The functions $f$ and $G$ are continuously differentiable on $\SD$.
	\end{enumerate}

	Let $h\colon\SD\to\BBR$ be continuously differentiable, and define the \emph{safe set} $\SC\triangleq\{x\in\SD\colon h(x)\geq 0\}$. Define $\dot h\colon\SD\times\SU\to\BBR$ by
	\begin{equation}
		\dot{h}(x,u)\triangleq L_f h(x)+L_G h(x)u,
		\label{eq:hdot}
	\end{equation}
	which is the derivative of $h$ along~\eqref{eq:dynamics}. For all $x\in\SD$, since $\SU$ is nonempty and compact and $u\mapsto\dot h(x,u)$ is continuous, $\max_{u\in\SU}\dot h(x,u)$ exists.
	
	\begin{definition}\label{def:cbf}
		The function $h$ is a \emph{CBF}~\citep{ames2017cbf} for~\eqref{eq:dynamics} on $\SD$ if there exists an extended class $\SK$ function $\alpha$ such that, for all $x\in\SD$,
		\begin{equation}
			\max_{u\in\SU}\dot h(x,u)\geq-\alpha\bigl(h(x)\bigr).
			\label{eq:cbf_condition}
		\end{equation}
	\end{definition}

	The following lemma restates the invariance result of~\citet[Cor.~2]{ames2017cbf}, which is given there without a separate proof. As in that result, the hypothesis is an inequality that holds for all $x\in\SD$, and not for $x\in\SC$ alone, and thus the proof, which Appendix~\ref{app:invariance} contains, requires no condition on $h'$ on $\{x\in\SD\colon h(x)=0\}$.
	
\begin{lemma}\label{lem:invariance}
	Consider~\eqref{eq:dynamics}, let $\alpha$ be an extended class $\SK$ function, and let $\mu\colon\SD\to\SU$ be locally Lipschitz continuous. Assume that~\ref{assum:dynamics} holds and that, for all $x\in\SD$, $\dot h(x,\mu(x))\ge-\alpha(h(x))$. Then, for all $x_0\in\SC$, the following statements hold:
		\begin{enumerate}[label={\it{\roman*}})]
			\item\label{lem1_p1} There exist $t_{\max}\in(0,\infty]$ and a unique maximal solution $s\colon[0,t_{\max})\to\SD$ of~\eqref{eq:dynamics} with $u=\mu\circ s$ and $s(0)=x_0$.
			\item\label{lem1_p2} For all $t\in[0,t_{\max})$, $s(t)\in\SC$.
	\end{enumerate}
	Thus, $\SC$ is \emph{forward invariant} with respect to~\eqref{eq:dynamics} under $\mu$.
\end{lemma}
	
A CBF-based safety filter \citep{ames2017cbf, ames2019cbf_qp} replaces a nominal controller $u_\mathrm{n}\colon\SD\to\BBR^m$ by the solution of a QP. Let $\alpha$ be an extended class $\SK$ function. For all $x\in\SD$ such that $\max_{u\in\SU}\dot h(x,u)\ge-\alpha(h(x))$, define
\begin{equation}
	\begin{aligned}
		u_*(x)\triangleq{}&\arg\min_{u\in\SU}\|u-u_\mathrm{n}(x)\|^2\\
		&\;\;\mathrm{s.t.}\;\;L_f h(x)+L_G h(x)u\geq-\alpha\bigl(h(x)\bigr),
	\end{aligned}
	\label{eq:qp}
\end{equation}
where the minimizer exists and is unique, because the objective is strictly convex and the set of all $u\in\SU$ that satisfy the constraint is nonempty, compact, and convex. Since the constraint of~\eqref{eq:qp} is affine in $u$, \eqref{eq:qp} is a convex QP if $\SU$ is a \emph{polyhedron}, which is the intersection of finitely many closed half-spaces of $\BBR^m$, and solvers such as qpOASES~\citep{ferreau2014qpoases} and OSQP~\citep{stellato2020osqp} solve such QPs.
	
	We parameterize the CBF by a feedforward neural network $h_\theta\colon\SD\to\BBR$. Let $d\in\BBZ_+$ be the network depth, and let $n_0,n_1,\dots,n_d\in\BBZ_+$ denote the layer widths, where $n_0\triangleq n$ and $n_d\triangleq 1$. The network is the composition
	\begin{equation}
		h_\theta \triangleq \ell_d \circ \sigma_{d-1} \circ \ell_{d-1} \circ \cdots \circ \sigma_1 \circ \ell_1,
		\label{eq:nn_def}
	\end{equation}
	where, for all $i\in\{1,\ldots,d\}$, the affine map $\ell_i\colon\BBR^{n_{i-1}}\to\BBR^{n_i}$ is defined, for all $z\in\BBR^{n_{i-1}}$, by $\ell_i(z)\triangleq W_i z+b_i$, with the weight matrix $W_i\in\BBR^{n_i\times n_{i-1}}$ and the bias vector $b_i\in\BBR^{n_i}$. For all $i\in\{1,\ldots,d-1\}$, the activation layer $\sigma_i\colon\BBR^{n_i}\to\BBR^{n_i}$ applies a scalar activation $\bar\sigma_i\colon\BBR\to\BBR$ componentwise, that is, for all $\zeta\in\BBR^{n_i}$,
	\begin{equation}
		\sigma_i(\zeta)\triangleq\matl\bar\sigma_i(\zeta_{(1)}) & \cdots & \bar\sigma_i(\zeta_{(n_i)})\matr^\rmT.
		\label{eq:componentwise_activation}
	\end{equation}
	The activations considered in this paper are defined, for all $t\in\BBR$, by the ReLU $\relu(t)\triangleq\max\{0,t\}$, which is differentiable on $\BBR\setminus\{0\}$, and by the hyperbolic tangent $\tanh t$, the sigmoid $1/(1+e^{-t})$, and the softplus $\log(1+e^{t})$, which are infinitely differentiable on $\BBR$. For all $i\in\{1,\ldots,d-1\}$ and all $\zeta\in\BBR^{n_i}$ such that, for all $j\in\{1,\ldots,n_i\}$, $\bar\sigma_i$ is differentiable at $\zeta_{(j)}$, it follows from~\eqref{eq:componentwise_activation} that $\sigma_i$ is differentiable at $\zeta$ and that its Jacobian is
	\begin{equation}
		\sigma_i'(\zeta)=\diag\bigl(\matl\bar\sigma_i'(\zeta_{(1)}) & \cdots & \bar\sigma_i'(\zeta_{(n_i)})\matr^\rmT\bigr)\in\BBR^{n_i\times n_i}.
		\label{eq:sigma_prime}
	\end{equation} The pointwise assumptions under which the compiler returns exact derivatives are stated in Lemma~\ref{lem:dual_activation}, Theorem~\ref{thm:main}, and Remark~\ref{rem:kinks}. The parameter vector $\theta\in\BBR^{n_\theta}$ collects the weights and biases as
	\begin{equation}
		\theta\triangleq\matl\mathrm{vec}(W_1)^\rmT & b_1^\rmT & \cdots & \mathrm{vec}(W_d)^\rmT & b_d^\rmT\matr^\rmT\in\BBR^{n_\theta},
		\label{eq:theta_stack}
	\end{equation}
	where the number of parameters is
	\begin{equation}
		n_\theta\triangleq\sum_{i=1}^{d}(n_i n_{i-1}+n_i).
		\label{eq_ntheta}
	\end{equation}

	With $h=h_\theta$, the QP~\eqref{eq:qp} requires, at each control cycle, the barrier value $h_\theta(x)$, the drift Lie derivative $L_f h_\theta(x)$, and the input Lie derivative $L_G h_\theta(x)$. To describe their computation, let $\sigma_d$ denote the identity map on $\BBR$, and, for all $x\in\SD$, define the \emph{forward sweep} by $\hat a_0(x)\triangleq x$ and, for all $i\in\{1,\ldots,d\}$,
\begin{align}
	a_i(x)&\triangleq\ell_i\bigl(\hat a_{i-1}(x)\bigr)\in\BBR^{n_i},\nn\\*
	\hat a_i(x)&\triangleq\sigma_i\bigl(a_i(x)\bigr)\in\BBR^{n_i},
	\label{eq:forward_sweep}
\end{align}
where $a_i(x)$ and $\hat a_i(x)$ are the \emph{preactivation} and the \emph{postactivation} of layer $i$, respectively. Thus, $h_\theta(x)=\hat a_d(x)$. We omit the argument $x$ of $a_i$ and $\hat a_i$ where $x$ is fixed. Let $\SD_\theta$ denote the set of all $x\in\SD$ such that, for all $i\in\{1,\ldots,d-1\}$ and all $j\in\{1,\ldots,n_i\}$, $\bar\sigma_i$ is differentiable at $a_{i(j)}(x)$.

We state reverse-mode AD, by which training frameworks compute the gradient of $h_\theta$, for later comparison. For all $i\in\{0,\ldots,d\}$, let $h_{\theta,i}\colon\BBR^{n_i}\to\BBR$ denote the map from the postactivation of layer $i$ to the output of the network. Specifically, $h_{\theta,d}$ is the identity map on $\BBR$, and, for all $i\in\{1,\ldots,d\}$, $h_{\theta,i-1}\triangleq h_{\theta,i}\circ\sigma_i\circ\ell_i$. Thus, for all $x\in\SD$ and all $i\in\{0,\ldots,d\}$, $h_\theta(x)=h_{\theta,i}(\hat a_i(x))$. Let $x\in\SD_\theta$. Since $h_{\theta,d}$ is differentiable and, for all $i\in\{1,\ldots,d\}$, $\sigma_i$ is differentiable at $a_i$, the chain rule implies that, for all $i\in\{0,\ldots,d\}$, $h_{\theta,i}$ is differentiable at $\hat a_i$. For all $i\in\{0,\ldots,d\}$, define the \emph{adjoint} $\delta_i\triangleq h_{\theta,i}'(\hat a_i)^\rmT\in\BBR^{n_i}$. Then, $\delta_d=1$, $\delta_0=h_\theta'(x)^\rmT$, and the chain rule implies that, for all $i\in\{1,\ldots,d\}$,
\begin{equation}
	\delta_{i-1} = W_i^\rmT\sigma_i'(a_i)\,\delta_i.
	\label{eq:backprop}
\end{equation}
The \emph{backward sweep}~\eqref{eq:backprop}, for $i=d,\ldots,1$, returns $\delta_0$.

For all $i\in\{1,\ldots,d-1\}$, the evaluation of~\eqref{eq:backprop} at layer $i$ requires the Jacobian $\sigma_i'(a_i)$, which the forward sweep determines at layer $i$ and the backward sweep uses after the layers $d-1,\ldots,i+1$. A reverse-mode implementation that completes the forward sweep before the backward sweep therefore holds quantities of all hidden layers at the same time, and its workspace grows with the sum of the hidden-layer widths (Proposition~\ref{prop:memory}).

The objective of this paper is to compute $h_\theta$, $L_f h_\theta$, and $L_G h_\theta$ by forward-mode differentiation, with a depth-independent workspace and without dynamic allocation, and to compare this computation with reverse-mode differentiation in arithmetic, memory, and execution time.

	\section{Dual Number Algebra}
	\label{sec:dual_algebra}
	
	This section states the algebra of forward-mode differentiation with dual numbers. In this arithmetic, each operation computes a derivative together with a value. The results of this section and of Section~\ref{sec:dual_nn} are standard facts of forward-mode differentiation \citep{griewank2008ad}. They are stated with the pointwise hypotheses that networks with the ReLU activation require, and they fix the notation of the complexity analysis and of the code.
	
	\begin{definition}\label{def:dual}
		The set of \emph{dual numbers} $\BBD$ is the two-dimensional real vector space with basis $\{1,\eps\}$, whose elements are written $z=a+b\eps$, where $a\in\BBR$ is the \emph{real part} and $b\in\BBR$ is the \emph{dual part} of $z$. For all $z_1\triangleq a_1+b_1\eps\in\BBD$ and $z_2\triangleq a_2+b_2\eps\in\BBD$, addition and multiplication are defined by
		\begin{align}
			z_1+z_2 &\triangleq (a_1+a_2)+(b_1+b_2)\eps,\label{eq:dual_add}\\*
			z_1 z_2 &\triangleq a_1 a_2+(a_1 b_2+a_2 b_1)\eps.\label{eq:dual_mul}
		\end{align}
	\end{definition}
	
	Since $1$ and $\eps$ are linearly independent, the real part and the dual part of each $z\in\BBD$ are unique. The real part operator $\Real\colon\BBD\to\BBR$ and the dual part operator $\Dual\colon\BBD\to\BBR$ are thus defined, for all $z=a+b\eps\in\BBD$, by $\Real(z)\triangleq a$ and $\Dual(z)\triangleq b$, which are the real and $\eps$ \emph{components} of $z$, respectively, and we identify $a\in\BBR$ with $a+0\eps$. Setting $a_1=a_2=0$ and $b_1=b_2=1$ in~\eqref{eq:dual_mul} shows that $\eps^2=0$, whereas $\eps\neq0$ because $\eps$ is a basis element.
	The following fact states that $\BBD$ is a commutative ring that is not a field. This structure was introduced by \citet{clifford1871dual} and is used in kinematics~\citep[Sec.~13.3]{bottema1990theoretical} and in AD~\citep[Sec.~3.1.1]{baydin2018ad_survey}.
	
	\begin{fact}\label{fact:ring}
		The triple $(\BBD,+,\cdot\,)$ is a commutative ring with additive identity $0+0\eps$ and multiplicative identity $1+0\eps$. This ring is not a field, because, for all $b\in\BBR\setminus\{0\}$, $b\eps\neq0$ and $(b\eps)^2=0$, and thus $b\eps$ has no multiplicative inverse.
	\end{fact}
	
	Forward-mode differentiation extends real functions to dual arguments such that the dual part of the value is a derivative~\citep{baydin2018ad_survey}.
	
	\begin{definition}\label{def:dual_ext}
		Let $I\subseteq\BBR$ be open, and let $\varphi\colon I\to\BBR$ be differentiable at $a\in I$. For all $b\in\BBR$, the \emph{dual extension} of $\varphi$ at $z=a+b\eps\in\BBD$ is
		\begin{equation}
			\tilde\varphi(z)\triangleq\varphi(a)+\varphi'(a)b\eps,
			\label{eq:dual_RD}
		\end{equation}
		and thus its real and $\eps$ components are $\varphi(a)$ and $\varphi'(a)b$, respectively.
	\end{definition}
	
	Proposition~\ref{thm:exact_diff} states the relation of the dual extension to the arithmetic of $\BBD$, namely, the sum, product, and chain rules of differentiation, and Appendix~\ref{app:thm_exact_diff} contains the proof.
	
	\begin{proposition}\label{thm:exact_diff}
		Let $I\subseteq\BBR$ be open, let $\varphi,\psi\colon I\to\BBR$ be differentiable at $a\in I$, let $b\in\BBR$, and let $z\triangleq a+b\eps\in\BBD$. Then, the following statements hold:
		\begin{enumerate}[label={\it{\roman*}})]
			\item\label{ted_p1} $\widetilde{(\varphi+\psi)}(z)=\tilde\varphi(z)+\tilde\psi(z)$.
			\item\label{ted_p2} $\widetilde{(\varphi\psi)}(z)=\tilde\varphi(z)\,\tilde\psi(z)$.
			\item\label{ted_p3} If $J\subseteq\BBR$ is open, $\varphi(I)\subseteq J$, and $\chi\colon J\to\BBR$ is differentiable at $\varphi(a)$, then $\widetilde{(\chi\circ\varphi)}(z)=\tilde\chi(\tilde\varphi(z))$.
			\item\label{ted_p4} If $r\in\BBN$ and $\gamma_0,\ldots,\gamma_r\in\BBR$ satisfy, for all $t\in I$, $\varphi(t)=\sum_{i=0}^{r}\gamma_it^i$, then $\tilde\varphi(z)=\sum_{i=0}^{r}\gamma_iz^i$, where $z^0\triangleq1$ and the powers are computed using~\eqref{eq:dual_mul}.
		\end{enumerate}
	\end{proposition}
	
	\begin{remark}\label{rem:taylor_exact}
		For a three times continuously differentiable function, a one-sided or central finite difference with step $\Delta>0$ has a truncation error of order $\Delta$ or $\Delta^2$, respectively. The dual part of~\eqref{eq:dual_RD} is the derivative $\varphi'(a)b$ and contains no truncation error. Proposition~\ref{thm:exact_diff} implies that a program that composes sums, products, and dual extensions of elementary functions returns the dual extension of the composite function. The components are computed in floating-point arithmetic and are thus subject to roundoff and to the errors of the library implementations of the elementary functions. Section~\ref{sec:examples} measures these errors.\exmark
	\end{remark}
	
	Sums and products of dual vectors $x+v\eps\in\BBD^n$, where $x,v\in\BBR^n$, and of dual matrices $A+B\eps\in\BBD^{n\times m}$, where $A,B\in\BBR^{n\times m}$, are defined by the matrix operations with the arithmetic of Definition~\ref{def:dual} for the entries, and $\Real$ and $\Dual$ act entrywise, with $\Real(A+B\eps)=A$ and $\Dual(A+B\eps)=B$. For all $W\in\BBR^{p\times n}$ and all $x,v\in\BBR^n$, $(W+0\eps)(x+v\eps)=Wx+Wv\eps$, which requires two real matrix--vector products.
	
	\begin{definition}\label{def:dual_vec}
		Let $\SD\subseteq\BBR^n$ be open, and let $\varphi\colon\SD\to\BBR^p$ be differentiable at $x\in\SD$. For all $v\in\BBR^n$, the \emph{dual extension} of $\varphi$ at $z=x+v\eps\in\BBD^n$ is $\tilde\varphi(z)\triangleq\varphi(x)+\varphi'(x)v\eps\in\BBD^p$.
	\end{definition}
	
	For $n=p=1$, Definition~\ref{def:dual_vec} coincides with Definition~\ref{def:dual_ext}. Corollary~\ref{cor:vec_ext}, which Appendix~\ref{app:thm_exact_diff} also proves, extends Proposition~\ref{thm:exact_diff}\ref{ted_p3} to vector functions.
	\begin{corollary}\label{cor:vec_ext}
		Let $\SD\subseteq\BBR^n$ and $\mathcal{E}\subseteq\BBR^p$ be open, let $\varphi\colon\SD\to\mathcal{E}$ be differentiable at $x\in\SD$, and let $\psi\colon\mathcal{E}\to\BBR^q$ be differentiable at $\varphi(x)$. Then, for all $v\in\BBR^n$, $z\triangleq x+v\eps$ satisfies
		\begin{equation}
			\tilde\psi\bigl(\tilde\varphi(z)\bigr)=\widetilde{(\psi\circ\varphi)}(z).
			\label{eq:vec_chain}
		\end{equation}
	\end{corollary}
	
	\section{Dual Neural Network Compiler}
	\label{sec:dual_nn}
	
	This section extends the layers of the network to dual arguments and shows that the evaluation of the extended network returns a Lie derivative in the dual part of its output. Appendix~\ref{app:dual_layers} contains the proofs of Lemmas~\ref{lem:dual_affine} and~\ref{lem:dual_activation}.
	
	\subsection{Dual Extension of Network Layers}
	
	\begin{lemma}\label{lem:dual_affine}
		Let $i\in\{1,\ldots,d\}$. Then, for all $a,c\in\BBR^{n_{i-1}}$, the dual extension of the affine layer $\ell_i$ at $z\triangleq a+c\eps$ satisfies
		\begin{equation}
			\tilde\ell_i(z) = (W_ia+b_i)+W_ic\eps.
			\label{eq:dual_affine}
		\end{equation}
	\end{lemma}
	
	\begin{remark}\label{rem:affine_cost}
		It follows from~\eqref{eq:dual_affine} that, for all $i\in\{1,\ldots,d\}$, the dual affine layer $\tilde\ell_i$ requires two matrix--vector products with the same matrix $W_i$ and one vector addition, and its dual part $W_ic$ does not depend on $b_i$. The dual pass therefore requires no parameter storage in addition to that of $\theta$.\exmark
	\end{remark}
	
	\begin{lemma}\label{lem:dual_activation}
		Let $i\in\{1,\ldots,d-1\}$, let $a,c\in\BBR^{n_i}$, and assume that, for all $j\in\{1,\ldots,n_i\}$, $\bar\sigma_i$ is differentiable at $a_{(j)}$. Then, the dual extension of the activation layer $\sigma_i$ at $z\triangleq a+c\eps$ satisfies
		\begin{equation}
			\tilde\sigma_i(z) = \sigma_i(a)+\sigma_i'(a)c\eps,
			\label{eq:dual_activation}
		\end{equation}
		where $\sigma_i'(a)$ is the diagonal matrix~\eqref{eq:sigma_prime}. Furthermore, for all $j\in\{1,\ldots,n_i\}$, the $j$th component of~\eqref{eq:dual_activation} is the scalar extension~\eqref{eq:dual_RD} of $\bar\sigma_i$ at $z_{(j)}$.
	\end{lemma}
	
	Let $i\in\{1,\ldots,d-1\}$. For the ReLU activation $\bar\sigma_i=\relu$, we define the dual activation layer, for all $a,c\in\BBR^{n_i}$ and all $j\in\{1,\ldots,n_i\}$, by
	\begin{equation}
		\bigl(\tilde\sigma_i(a+c\eps)\bigr)_{(j)}=\begin{cases}a_{(j)}+c_{(j)}\eps, & a_{(j)}>0,\\ 0+0\eps, & a_{(j)}\leq 0. \end{cases}
		\label{eq:dual_relu}
	\end{equation}
	If, for all $j\in\{1,\ldots,n_i\}$, $a_{(j)}\neq0$, then Lemma~\ref{lem:dual_activation} implies that~\eqref{eq:dual_relu} is the dual extension of the ReLU layer at $a+c\eps$. If $j\in\{1,\ldots,n_i\}$ and $a_{(j)}=0$, then the $j$th component of~\eqref{eq:dual_relu} corresponds to the convention $\relu'(0)\triangleq0$ (Remark~\ref{rem:kinks}). The branch of~\eqref{eq:dual_relu} is selected by $a_{(j)}$ and does not depend on $c_{(j)}$. For the hyperbolic tangent, the sigmoid, and the softplus activations, the dual layers follow from~\eqref{eq:dual_activation} with, for all $t\in\BBR$, $\bar\sigma_i'(t)=1-\bar\sigma_i(t)^2$, $\bar\sigma_i'(t)=\bar\sigma_i(t)(1-\bar\sigma_i(t))$, and $\bar\sigma_i'(t)=1/(1+e^{-t})$, respectively.
	
	\subsection{Main Result: Exact Lie Derivatives}
	
	A \emph{seed} is a dual vector of the form $x+\xi(x)\eps\in\BBD^n$, where $x\in\SD$ and $\xi\colon\SD\to\BBR^n$ is a vector field. The following theorem states that the evaluation of the dual-extended network at a seed returns the Lie derivative along $\xi$ in the dual part of the output.
	
	\begin{theorem}\label{thm:main}
		Consider the neural network $h_\theta$ defined in~\eqref{eq:nn_def}, let $x\in\SD_\theta$, define the \emph{dual-extended network}
		\begin{equation}
			\tilde h_\theta\triangleq\tilde\ell_d\circ\tilde\sigma_{d-1}\circ\cdots\circ\tilde\sigma_1\circ\tilde\ell_1,
			\label{eq:h_tilde_def}
		\end{equation}
		and, for all $\xi\colon\SD\to\BBR^n$, define
		\begin{equation}
			x_\xi\triangleq x+\xi(x)\eps\in\BBD^n.
			\label{eq:seed_dual}
		\end{equation}
		Then, $h_\theta$ is differentiable at $x$, and, for all $\xi\colon\SD\to\BBR^n$,
		\begin{equation}
			\tilde h_\theta(x_\xi) = h_\theta(x)+L_\xi h_\theta(x)\eps.
			\label{eq:main_result}
		\end{equation}
		In particular, $\Dual(\tilde h_\theta(x_f))=L_fh_\theta(x)$, and, for all $j\in\{1,\ldots,m\}$, $x_{G,j}\triangleq x_{\col_j(G)}$ satisfies
		\begin{equation}
			\Dual\bigl(\tilde h_\theta(x_{G,j})\bigr)=\col_j\bigl(L_G h_\theta(x)\bigr).\label{eq_DuhG}
		\end{equation}
	\end{theorem}
	
	\begin{proof}
		Let $\xi\colon\SD\to\BBR^n$, let $\tilde\sigma_d$ denote the identity map on $\BBD$, define $z_0\triangleq x_\xi$, and, for all $i\in\{1,\ldots,d\}$, define $z_i\triangleq\tilde\sigma_i(\tilde\ell_i(z_{i-1}))$. Then, $z_d=\tilde h_\theta(x_\xi)$. For induction, let $\kappa\in\{0,\ldots,d-1\}$ be such that $\hat a_\kappa$ is differentiable at $x$ and $z_\kappa$ is the dual extension of $\hat a_\kappa$ at $x_\xi$, which holds for $\kappa=0$ because $\hat a_0$ is the identity map on $\SD$. The affine layer $\ell_{\kappa+1}$ is differentiable at $\hat a_\kappa(x)$. Since $x\in\SD_\theta$ and $\sigma_d$ is the identity map on $\BBR$, the activation layer $\sigma_{\kappa+1}$ is differentiable at $a_{\kappa+1}(x)=\ell_{\kappa+1}(\hat a_\kappa(x))$. Corollary~\ref{cor:vec_ext}, applied to $\ell_{\kappa+1}\circ\hat a_\kappa$ and then to $\sigma_{\kappa+1}\circ(\ell_{\kappa+1}\circ\hat a_\kappa)$, thus implies that $\hat a_{\kappa+1}$ is differentiable at $x$ and that $z_{\kappa+1}$ is the dual extension of $\hat a_{\kappa+1}$ at $x_\xi$, which confirms the claim for $\kappa+1$. Since the claim holds for $\kappa=d$, $\hat a_d=h_\theta$, and $h_\theta'(x)\xi(x)=L_\xi h_\theta(x)$, Definition~\ref{def:dual_vec} yields
		\begin{align}
			\tilde h_\theta(x_\xi) &= h_\theta(x)+h_\theta'(x)\xi(x)\eps \nn\\*
			&= h_\theta(x)+L_\xi h_\theta(x)\eps,
			\label{eq:iterated_chain}
		\end{align}
		which is~\eqref{eq:main_result}. Setting $\xi=f$ in~\eqref{eq:main_result} yields $\Dual(\tilde h_\theta(x_f))=L_fh_\theta(x)$. For all $j\in\{1,\ldots,m\}$, setting $\xi=\col_j(G)$ in~\eqref{eq:main_result} yields~\eqref{eq_DuhG}, because $L_{\col_j(G)}h_\theta(x)=\col_j(L_Gh_\theta(x))$.
	\end{proof}
	
	\begin{remark}\label{rem:kinks}
		For the hyperbolic tangent, the sigmoid, and the softplus activations, $\SD_\theta=\SD$. For networks with the ReLU activation, $\SD\setminus\SD_\theta=\bigcup_{i\in\{1,\ldots,d-1\}}\bigcup_{j\in\{1,\ldots,n_i\}}\{x\in\SD\colon a_{i(j)}(x)=0\}$, which is closed in $\SD$ because the functions $a_{i(j)}$ are continuous. At the states of this set, $h_\theta$ need not be differentiable, and the convention $\relu'(0)\triangleq 0$ selects the element $0$ of the Clarke generalized derivative $[0,1]$ of the ReLU at zero~\citep{clarke1990optimization}. PyTorch and TensorFlow apply the same convention in training, and the compiled code and the training framework therefore return the same values at these states. Theorem~\ref{thm:main} makes no statement at these states, where the value that the convention yields need not be an element of the Clarke generalized gradient of $h_\theta$. For example, the convention yields $0$ for the identity map $t\mapsto\relu(t)-\relu(-t)$ at $t=0$. In floating-point arithmetic, a preactivation within roundoff of zero can yield a computed sign that differs from the exact sign, and the returned value is then the derivative of the adjacent affine piece. Section~\ref{sec:examples} counts the sampled states at which a preactivation is within $10^{-4}$ of zero and reports their error separately.\exmark
	\end{remark}
	
	\section{Algorithm and Complexity}
\label{sec:algorithm}

This section presents the algorithm that assembles the CBF constraint, derives its operation count and its workspace, and shows that the dual pass evaluates the chain rule in the order of the forward sweep. Appendix~\ref{app:prop_dual_cost} contains the proofs of Theorem~\ref{thm:correctness}, Corollary~\ref{thm:safety}, and Proposition~\ref{prop:dual_cost}.

\subsection{Batched Dual Forward-Pass Algorithm}
	
	Theorem~\ref{thm:main} implies that one dual pass returns one input Lie derivative, and thus the $m$ input Lie derivatives require $m$ passes. Algorithm~\ref{alg:dual_cbf} collects the $m$ input seeds as the columns of one dual matrix.
	
	For all $x\in\SD$, let $X_G\triangleq x\mathbf{1}_m^\rmT+G(x)\eps\in\BBD^{n\times m}$, whose columns are the seeds $x_{G,1},\ldots,x_{G,m}$. For all $i\in\{1,\ldots,d\}$, all $Z\in\BBD^{n_{i-1}\times m}$, and all $j\in\{1,\ldots,m\}$, define $\col_j(\tilde\ell_i(Z))\triangleq\tilde\ell_i(\col_j(Z))$, and thus $\tilde\ell_i(Z)=W_iZ+b_i\mathbf{1}_m^\rmT$. For all $i\in\{1,\ldots,d-1\}$, all $Z\in\BBD^{n_i\times m}$, and all $j\in\{1,\ldots,m\}$, define $\col_j(\tilde\sigma_i(Z))\triangleq\tilde\sigma_i(\col_j(Z))$. Thus, for all $j\in\{1,\ldots,m\}$, the $j$th column of $\tilde h_\theta(X_G)$ is $\tilde h_\theta(x_{G,j})$. For all $x\in\SD_\theta$, Theorem~\ref{thm:main}, applied to each column, yields
	\begin{equation}
		\tilde h_\theta(X_G) = h_\theta(x)\mathbf{1}_m^\rmT+L_G h_\theta(x)\eps.
		\label{eq:batched_eval}
	\end{equation}
	The matrix $X_G$ is thus a notation for the $m$ seeds, and~\eqref{eq:batched_eval} holds for all orders in which the columns are evaluated.
	
	Algorithm~\ref{alg:dual_cbf} assembles the CBF constraint of~\eqref{eq:qp} with a vector evaluation at $x_f$, which returns $h_\theta(x)$ and $L_f h_\theta(x)$, and a matrix evaluation at $X_G$, which returns $L_G h_\theta(x)$. For these evaluations, the compiler emits two \emph{entry points} (i.e., functions that the user code calls). The \emph{sequential entry point} evaluates the seeds $x_f$, $x_{G,1},\ldots,x_{G,m}$ one at a time in one buffer, and its workspace is that of Proposition~\ref{prop:memory}, which does not depend on $m$. The \emph{shared-primal entry point} evaluates $\tilde h_\theta$ at $x\mathbf{1}_{m+1}^\rmT+[f(x)\ \ G(x)]\eps\in\BBD^{n\times(m+1)}$ in one sweep and stores the real part, which is the same in all $m+1$ columns, once. Relative to the sequential entry point, it requires fewer operations and a larger workspace (Remark~\ref{rem:ad_cost}), and the choice between the two entry points is thus a trade-off between memory and arithmetic. Neither entry point uses the vector instructions of the processor.
	
	\begin{algorithm}[ht]
		\caption{Batched Dual-Algebraic CBF Assembly}
		\label{alg:dual_cbf}
		\begin{algorithmic}[1]
			\REQUIRE $x\in\SD$, $f(x)$, $G(x)$, $u_\mathrm{n}(x)$, compact convex $\SU$, weights $\theta$, extended class $\SK$ function $\alpha$
			\ENSURE Barrier value $h$, $L_f$, $L_G\in\BBR^{1\times m}$, control $u_*$, feasibility of~\eqref{eq:qp}
			\STATE $y_f \leftarrow \tilde h_\theta(x+f(x)\eps)$ \COMMENT{Vector evaluation}
			\STATE $h \leftarrow \Real(y_f)$;~$L_f \leftarrow \Dual(y_f)$
			\STATE $X_G \leftarrow x\mathbf{1}_m^\rmT+G(x)\eps$
			\STATE $Y_G \leftarrow \tilde h_\theta(X_G)$ \COMMENT{Matrix evaluation}
			\STATE $L_G \leftarrow \Dual(Y_G)$
			\STATE Formulate constraint $L_f+L_G u\geq-\alpha(h)$
			\IF{$\max_{u\in\SU}(L_f+L_Gu)\ge-\alpha(h)$}
			\STATE Solve~\eqref{eq:qp} to obtain $u_*$ \COMMENT{Filtered control}
			\ELSE
			\STATE $u_*\leftarrow\arg\min_{u\in\arg\max_{v\in\SU}L_Gv}\|u-u_\mathrm{n}(x)\|^2$ \COMMENT{Fallback}
			\ENDIF
			\RETURN $h$, $L_f$, $L_G$, $u_*$, and the feasibility of~\eqref{eq:qp}
		\end{algorithmic}
	\end{algorithm}
	
	\begin{theorem}\label{thm:correctness}
		Consider the neural network $h_\theta$ defined in~\eqref{eq:nn_def}, and let $x\in\SD_\theta$. Then, Algorithm~\ref{alg:dual_cbf} returns $h=h_\theta(x)$, $L_f=L_f h_\theta(x)$, and $L_G=L_G h_\theta(x)$. In particular, the constraint assembled at line~6 is the constraint of~\eqref{eq:qp}.
	\end{theorem}
	
	The following corollary combines Theorem~\ref{thm:correctness} with the invariance result of~\citet{ames2017cbf}. It identifies the assumptions on training and control design under which the compiled filter yields forward invariance.
	
	\begin{corollary}\label{thm:safety}
	Consider the neural network $h_\theta$ defined in~\eqref{eq:nn_def}, let $\alpha$ be an extended class $\SK$ function, and assume that~\ref{assum:dynamics} and the following conditions hold:
	\begin{enumerate}[label={\it{\roman*}})]
		\item\label{cor2_c1} $\SU$ is a polyhedron.
		\item\label{cor2_c2} For all $i\in\{1,\ldots,d-1\}$, $\bar\sigma_i$ is continuously differentiable.
		\item\label{cor2_c3} For all $x\in\SD$, \eqref{eq:cbf_condition} holds with $h=h_\theta$, and thus $h_\theta$ is a CBF for~\eqref{eq:dynamics} on $\SD$.
	\end{enumerate}
	Then, for all $x\in\SD$, Algorithm~\ref{alg:dual_cbf} executes line~8, and the QP solved at line~8 is~\eqref{eq:qp} with $h=h_\theta$ and has a unique solution $u_*(x)$. If, in addition, $u_*\colon\SD\to\SU$ is locally Lipschitz continuous, then $\{x\in\SD\colon h_\theta(x)\ge0\}$ is forward invariant with respect to~\eqref{eq:dynamics} under $u_*$.
\end{corollary}

	\begin{remark}\label{rem:scope}
		Corollary~\ref{thm:safety} contains no new invariance result, because its invariance argument relies on Lemma~\ref{lem:invariance}. The corollary states that, in exact arithmetic, the compiled constraint coincides with the constraint that the theory analyzes. Hypothesis~\ref{cor2_c2} holds for the hyperbolic tangent, the sigmoid, and the softplus activations and does not hold for the ReLU. Corollary~\ref{thm:safety} therefore makes no statement about the ReLU networks of Experiments~\ref{ex:bicycle} and~\ref{ex:neural}, which are computational demonstrations, and a corresponding result for ReLU networks requires a nonsmooth analysis, which is not part of this paper. Hypothesis~\ref{cor2_c3} is a property of the training or verification procedure~\citep{dawson2023safe, so2023ncbf}, which the compiler does not address. At a state where the constraint cannot be satisfied within $\SU$, Algorithm~\ref{alg:dual_cbf} returns the fallback input of line~10, which violates the constraint, and hypothesis~\ref{cor2_c3} excludes such states. Without input bounds, $u_*$ is locally Lipschitz continuous under the conditions stated in~\citet{ames2017cbf}, which do not cover the compact set $\SU$ of~\eqref{eq:dynamics}. This paper does not analyze inter-sample effects.\exmark
	\end{remark}
	
	\subsection{Computational and Memory Complexity}
	
	Throughout this subsection, each floating-point multiplication, addition, subtraction, and scalar activation evaluation counts as one operation. This count is a model of the arithmetic and not of the execution time. On a microcontroller, the ReLU is a comparison, whereas the hyperbolic tangent, the sigmoid, and the softplus call the library routines \texttt{tanhf}, \texttt{expf}, and \texttt{logf}, whose execution time exceeds that of one multiplication and depends on the argument. The counts therefore compare the arithmetic of methods on the same network, and Section~\ref{sec:examples} reports the measured execution times, which include the library calls. For all $i\in\{1,\ldots,d-1\}$, let $c_{\sigma_i}$ denote the number of operations required to evaluate the scalar derivative $\bar\sigma_i'$ given the value of $\bar\sigma_i$ at the same argument. For the ReLU activation, $c_{\sigma_i}=0$, because the derivative is determined by the sign of the preactivation, which is available from the activation, and the selection of the dual part in~\eqref{eq:dual_relu} counts as one operation per neuron, as the product with $\sigma_i'(a_i)$ does. For the hyperbolic tangent and the sigmoid, the identities $\bar\sigma_i'=1-\bar\sigma_i^2$ and $\bar\sigma_i'=\bar\sigma_i(1-\bar\sigma_i)$ each require one multiplication and one subtraction, and thus $c_{\sigma_i}=2$. For the softplus, $c_{\sigma_i}=1$, because $\bar\sigma_i'$ is the sigmoid, whose evaluation counts as one operation.
	
	Proposition~\ref{prop:dual_cost} states the number of operations of one dual pass and of the assembly of the CBF constraint.
	
	\begin{proposition}\label{prop:dual_cost}
		Consider the neural network $h_\theta$ defined in~\eqref{eq:nn_def}, and let
		\begin{equation}
			C_\mathrm{f}\triangleq\sum_{i=1}^d 2n_i n_{i-1}+\sum_{i=1}^{d-1}n_i
			\label{eq:Cf_def}
		\end{equation}
		denote the number of operations of the forward sweep~\eqref{eq:forward_sweep}. Then, one dual pass requires
		\begin{equation}
			C_\eps\triangleq 2C_\mathrm{f}+\sum_{i=1}^{d-1}c_{\sigma_i}n_i-\sum_{i=1}^{d}n_i
			\label{eq:Cdf_def}
		\end{equation}
		operations, and the computation of the drift and input Lie derivatives by Algorithm~\ref{alg:dual_cbf} requires $(m+1)C_\eps$ operations.
	\end{proposition}
	
	Forward mode and reverse mode both read the parameter vector $\theta$, whose $n_\theta$ floats reside in \emph{read-only memory} (i.e., memory that is written when the program is loaded and not while it runs). The following comparison therefore concerns the workspace, which is the \emph{mutable memory} (i.e., memory that the code writes while it runs) that one evaluation of the constraint requires in addition to the parameters, the inputs $x$, $f(x)$, and $G(x)$, and the outputs $h$, $L_f$, and $L_G$. Memory is counted in floats, which have single precision in Section~\ref{sec:examples}.
	
	\begin{proposition}\label{prop:memory}
	Consider the neural network $h_\theta$ defined in~\eqref{eq:nn_def}. Then, the following statements hold:
	\begin{enumerate}[label={\it{\roman*}})]
		\item\label{prop2_p1} For all implementations of reverse mode that complete the forward sweep~\eqref{eq:forward_sweep} before the backward sweep~\eqref{eq:backprop} begins and that evaluate no part of~\eqref{eq:forward_sweep} during the backward sweep, the workspace is at least $w_\mathrm{r}\triangleq\sum_{i=1}^{d-1}s_i$ floats, where, for all $i\in\{1,\ldots,d-1\}$, $s_i$ is the memory, in floats, of the quantities that the forward sweep stores at layer $i$ and from which the backward sweep computes $\sigma_i'(a_i)$. If, in addition, for all $i\in\{1,\ldots,d-1\}$, these quantities are the components of $a_i$ or of $\hat a_i$, then $w_\mathrm{r}=\sum_{i=1}^{d-1}n_i$.
		\item\label{prop2_p2} With the $m+1$ seeds evaluated one at a time, lines~\mbox{1--5} of Algorithm~\ref{alg:dual_cbf} admit an implementation whose workspace consists of $w_\eps\triangleq 2\max_{i\in\{1,\ldots,d\}}(n_{i-1}+n_i)$ floats and of scalar variables whose number does not depend on $d$, $m$, or $n_0,\ldots,n_d$. Furthermore, $w_\eps\le 4\max_{i\in\{0,\ldots,d\}}n_i$.
	\end{enumerate}
\end{proposition}

\begin{proof}
	To prove~\ref{prop2_p1}, let $i\in\{1,\ldots,d-1\}$. The evaluation of~\eqref{eq:backprop} at layer $i$ requires $\sigma_i'(a_i)$. Since the implementation evaluates no part of~\eqref{eq:forward_sweep} during the backward sweep, it computes $\sigma_i'(a_i)$ from the quantities that the forward sweep stores at layer $i$, which occupy $s_i$ floats. Since the forward sweep is complete before the backward sweep begins, these quantities are held at the beginning of the backward sweep. Thus, at the beginning of the backward sweep, the workspace holds the stored quantities of all $d-1$ hidden layers, which occupy $\sum_{i=1}^{d-1}s_i=w_\mathrm{r}$ floats. For the case where, for all $i\in\{1,\ldots,d-1\}$, these quantities are the $n_i$ components of $a_i$ or of $\hat a_i$, it follows that, for all $i\in\{1,\ldots,d-1\}$, $s_i=n_i$.
	
	To prove~\ref{prop2_p2}, let $x\in\SD$, let $\xi$ be one of the vector fields $f$, $\col_1(G)$, \ldots, $\col_m(G)$, and consider the dual states $z_0,\ldots,z_d$ of the proof of Theorem~\ref{thm:main}, where $z_0=x+\xi(x)\eps$, $z_d=\tilde h_\theta(z_0)$, and, for all $i\in\{0,\ldots,d\}$, $z_i\in\BBD^{n_i}$ occupies $2n_i$ floats. Consider the implementation that, for $i=1,\ldots,d$, computes the components of $z_i=\tilde\sigma_i(\tilde\ell_i(z_{i-1}))$ one at a time from $z_{i-1}$ and $\theta$ and releases $z_{i-1}$ after the last component. The sums and the activation of each component are accumulated in the two floats of this component and in scalar variables whose number does not depend on $d$, $m$, or $n_0,\ldots,n_d$. For all $i\in\{1,\ldots,d\}$, step $i$ thus holds the $2n_{i-1}$ floats of $z_{i-1}$, the $2n_i$ floats of $z_i$, and these scalar variables, and the maximum of $2(n_{i-1}+n_i)$ over the steps is $w_\eps$. Since the seeds are evaluated one at a time, and the components of $z_d$ are copied to the outputs after each evaluation, all seeds use the same memory. The inequality follows from $n_{i-1}+n_i\le 2\max_{j\in\{0,\ldots,d\}}n_j$.
\end{proof}
	
	\begin{remark}\label{rem:static_reverse}
		For a network of fixed topology, reverse-mode code can be generated ahead of time in the same way as the forward-mode code. The stored quantities of Proposition~\ref{prop:memory}\ref{prop2_p1}, which constitute the \emph{cache}, can then be placed in statically allocated buffers, without heap allocation and without a graph constructed at run time. With ahead-of-time code for both modes, the workspaces of the two modes differ in their dependence on the depth. If the stored quantities are the activations, then the lower bound is $w_\mathrm{r}=\sum_{i=1}^{d-1}n_i$, which increases with each hidden layer, whereas, for all $d\in\BBZ_+$, $w_\eps\le4\max_{i\in\{0,\ldots,d\}}n_i$. For networks with $d-1\ge2$ hidden layers of equal width $\bar n\ge n$, as in Section~\ref{sec:examples}, $w_\mathrm{r}=(d-1)\bar n$ and $w_\eps=4\bar n$, and thus $w_\eps<w_\mathrm{r}$ if and only if $d-1>4$. Three qualifications apply to part~\ref{prop2_p1}. First, checkpointing and recomputation \citep{griewank2000revolve} evaluate parts of the forward sweep again during the backward sweep, which reduces the cache and requires additional operations, and part~\ref{prop2_p1} does not cover them. Second, for all $i\in\{1,\ldots,d-1\}$ such that $\bar\sigma_i$ is the ReLU, the Jacobian $\sigma_i'(a_i)$ is determined by the signs of the components of $a_i$, and thus the stored quantities can be reduced to $n_i$ bits. Third, an implementation of~\eqref{eq:backprop} holds, in addition, the adjoint vectors. The statically allocated implementation of Section~\ref{sec:examples} stores the postactivations of all hidden layers and two adjoint buffers, which occupy $\sum_{i=1}^{d-1}n_i+2\max_{i\in\{1,\ldots,d-1\}}n_i$ floats. For equal widths, these buffers occupy $(d+1)\bar n$ floats, which exceeds $w_\eps$ if and only if $d-1>2$. The allocation of the graph at each call is a property of general-purpose AD frameworks, which record the operations of each call and do not assume a fixed topology, and Section~\ref{sec:examples} measures such an implementation as a secondary comparison.\exmark
	\end{remark}
	
	The following corollary follows from Proposition~\ref{prop:memory}.
	
	\begin{corollary}\label{cor:static_memory}
		With the $m+1$ seeds evaluated one at a time, lines~\mbox{1--5} of Algorithm~\ref{alg:dual_cbf} admit an implementation whose workspace consists of one buffer of $w_\eps$ floats, whose size and location are fixed before the program runs, and of scalar variables whose number does not depend on $d$, $m$, or $n_0,\ldots,n_d$. This implementation makes no dynamic allocation.
	\end{corollary}
	
	\begin{proof}
		Consider the implementation in the proof of Proposition~\ref{prop:memory}\ref{prop2_p2}, and let $B$ denote a buffer of $w_\eps$ floats. For all even $i\in\{0,\ldots,d\}$, store the $2n_i$ floats of $z_i$ in the first $2n_i$ floats of $B$, and, for all odd $i\in\{0,\ldots,d\}$, store them in the last $2n_i$ floats of $B$. Let $i\in\{1,\ldots,d\}$. Since $2n_{i-1}+2n_i\le w_\eps$, step $i$ writes $z_i$ without overwriting $z_{i-1}$, and, for $i\ge2$, it overwrites floats that hold $z_{i-2}$, which is not used after step $i-1$. Since $w_\eps$ is determined by $n_0,\ldots,n_d$, which are known when the program is compiled, the size and the location of $B$ are fixed before the program runs. Since the number of scalar variables of each step is fixed, these variables are local variables of the code, and no dynamic allocation occurs.
	\end{proof}
	
	For all $i\in\{1,\ldots,d\}$, the emitted code places $z_{i-1}$ in a statically allocated array, called a \emph{static array}, and $z_i$ in an \emph{automatic array}, which exists only while the code of layer $i$ runs and resides on the \emph{stack} (i.e., the memory of local variables), and Section~\ref{sec:examples} relates this placement to $w_\eps$. The emitted code contains no statements that allocate memory, which inspection of the generated header confirms. This inspection does not cover the routines of the C library that the header calls, namely, \texttt{expf}, \texttt{logf}, and \texttt{tanhf}. Section~\ref{sec:examples} therefore counts the allocations of the emitted code and of these routines on the target with \emph{allocation hooks} (i.e., functions that the system calls at each allocation).
	
	\begin{remark}\label{rem:ad_cost}
		Reverse mode computes the gradient $h_\theta'(x)^\rmT\in\BBR^n$ by the backward recurrence~\eqref{eq:backprop}. The vector--Jacobian products with the matrices $W_i^\rmT$ involve no bias additions and require $\sum_{i=1}^d(2n_i n_{i-1}-n_{i-1})$ operations. For all $i\in\{1,\ldots,d-1\}$, the evaluation of the activation derivatives of layer $i$ and the scaling by them require $(c_{\sigma_i}+1)n_i$ operations, and the forward sweep requires $C_\mathrm{f}$ operations. The gradient thus requires $2C_\mathrm{f}+\sum_{i=1}^{d-1}c_{\sigma_i}n_i-\sum_{i=1}^d n_{i-1}$ operations. Since $\sum_{i=1}^{d}(n_i-n_{i-1})=1-n$, it follows from~\eqref{eq:Cdf_def} that this number is $C_\eps-(n-1)$, and thus reverse mode computes the gradient with $n-1$ fewer operations than one dual pass. The products $h_\theta'(x)f(x)$ and $h_\theta'(x)G(x)$ require $(m+1)(2n-1)$ operations, and thus the reverse-mode assembly of the constraint requires
		\begin{equation}
			C_\mathrm{r}\triangleq C_\eps-(n-1)+(m+1)(2n-1)
			\label{eq:Cr_def}
		\end{equation}
		operations, whereas the dual passes require $(m+1)C_\eps$ operations. For fixed $n$, $m$, and $d\ge2$, $(m+1)C_\eps/C_\mathrm{r}\to m+1$ as $\min_{i\in\{1,\ldots,d-1\}}n_i\to\infty$. The $m+1$ dual passes compute the same real part. If the real part and the activation derivatives are computed once and the $m+1$ dual parts are propagated together with them in one sweep, then the number of operations is
		\begin{equation}
			C_\mathrm{s}\triangleq C_\mathrm{f}+\sum_{i=1}^{d-1}c_{\sigma_i}n_i+(m+1)\biggl(C_\mathrm{f}-\sum_{i=1}^{d}n_i\biggr),
			\label{eq:Csh_def}
		\end{equation}
		and the workspace is $(m+2)\max_{i\in\{1,\ldots,d\}}(n_{i-1}+n_i)$ floats. For fixed $n$, $m$, and $d\ge3$, $C_\mathrm{s}/C_\mathrm{r}\to(m+2)/2$ as $\min_{i\in\{1,\ldots,d-1\}}n_i\to\infty$. Since $(m+1)C_\eps-C_\mathrm{s}=m(C_\mathrm{f}+\sum_{i=1}^{d-1}c_{\sigma_i}n_i)$, the shared-primal sweep saves the operations of $m$ of the $m+1$ evaluations of the real part and of the activation derivatives, and it requires $m\max_{i\in\{1,\ldots,d\}}(n_{i-1}+n_i)$ floats in addition to $w_\eps$.\exmark
	\end{remark}

	\begin{example}\label{ex:deep}
		Let $n=4$ and $m=1$, let $d\ge3$, and, for all $i\in\{1,\ldots,d-1\}$, let $n_i=64$ and let $\bar\sigma_i$ be the hyperbolic tangent. Assume that $8$~KB of memory is available for the workspace, which holds $2{,}048$ floats in single precision. Proposition~\ref{prop:memory} yields $w_\mathrm{r}=64(d-1)$ floats for a reverse-mode implementation that stores the activations, which thus fits this memory only if $d\le33$. A network with $d>33$ requires more memory or recomputation (Remark~\ref{rem:static_reverse}). In contrast, for all $d\ge3$, $w_\eps=256$ floats, and the workspace of the shared-primal sweep is $384$ floats (Remark~\ref{rem:ad_cost}), which are one eighth and three sixteenths of the available memory, respectively.

		In exchange for a workspace that does not depend on the depth, the dual passes require more operations. It follows from~\eqref{eq:Cr_def} that $C_\mathrm{r}=C_\eps+11$, and thus, for all $d\ge3$, the two dual passes require fewer than twice the operations of reverse mode. For example, let $d=129$, in which case $w_\mathrm{r}=8{,}192$ floats, which is four times the available memory. Then, Proposition~\ref{prop:dual_cost} and~\eqref{eq:Csh_def} yield $C_\eps=2{,}106{,}623$ and $C_\mathrm{s}=3{,}147{,}646$, and thus the two dual passes and the shared-primal sweep require $2.00$ and $1.49$ times the operations of reverse mode, respectively.

		Networks of this depth arise where a map is applied repeatedly. Let $k\in\BBZ_+$, and let $\nu\colon\BBR^n\to\BBR^n$ and $\varrho\colon\BBR^n\to\BBR$ be feedforward networks with affine output layers and with $p$ and $q$ hidden layers, respectively. Since compositions of affine maps are affine, $h_\theta\triangleq\varrho\circ\nu^k$, where $\nu^k$ is the $k$-fold composition of $\nu$, has the form~\eqref{eq:nn_def} with $d=kp+q+1$, and thus Propositions~\ref{prop:dual_cost} and~\ref{prop:memory} apply, whereas the number of weights of $\nu$ and $\varrho$, which determine $\theta$, does not depend on $k$, although the emitted header stores the weights of all $d$ layers. For example, backup CBFs \citep{gurriet2018online,chen2021backup} compose the flow under a \emph{backup control law} over a fixed horizon with a function that is nonnegative on a \emph{backup set} and negative elsewhere, which $\nu^k$ and $\varrho$ can approximate, respectively. Repeated maps also arise as the time steps of recurrent networks \citep[Ch.~10]{goodfellow2016deep}, the integration steps of neural ordinary differential equations \citep{chen2018neural}, and the iterations of unrolled algorithms \citep{monga2021unrolling}, where a map that adds its input to its output, such as an integration step, is not of the form assumed for $\nu$.\exmark
	\end{example}

	\subsection{Layer-by-Layer Algebraic Trace}
	\label{sec:trace}

		Let $x\in\SD_\theta$, let $v\in\BBR^n$, and omit the argument $x$ of the preactivations $a_i$ and the postactivations $\hat a_i$ of Section~\ref{sec:problem}. Define $z_0\triangleq x+v\eps$ and $\hat c_0\triangleq v$. For all $i\in\{1,\ldots,d\}$, let $y_i\triangleq\tilde\ell_i(z_{i-1})$ and $c_i\triangleq\Dual(y_i)$, for all $i\in\{1,\ldots,d-1\}$, let $z_i\triangleq\tilde\sigma_i(y_i)$ and $\hat c_i\triangleq\Dual(z_i)$, and let $z_d\triangleq y_d$. Lemmas~\ref{lem:dual_affine} and~\ref{lem:dual_activation}, applied in the order of the layers, imply that, for all $i\in\{1,\ldots,d\}$, $\Real(y_i)=a_i$ and
	\begin{equation}
		c_i = W_i\hat c_{i-1},\label{eq:dual_affine_trace}
	\end{equation}
	and that, for all $i\in\{1,\ldots,d-1\}$, $\Real(z_i)=\hat a_i$ and
	\begin{equation}
		\hat c_i = \sigma_i'(a_i)c_i.\label{eq:dual_act_trace}
	\end{equation}
		Since $\hat c_0=v$,~\eqref{eq:dual_affine_trace} and~\eqref{eq:dual_act_trace} imply that $c_1=W_1v$ and, for all $i\in\{2,\ldots,d\}$, $c_i=W_i\sigma_{i-1}'(a_{i-1})c_{i-1}$. Since, in addition, $y_d=\tilde h_\theta(x+v\eps)$, Theorem~\ref{thm:main} with $\xi\equiv v$ implies that
	\begin{align}
		c_d &= W_d\sigma_{d-1}'(a_{d-1})W_{d-1}\cdots\sigma_1'(a_1)W_1 v \nn\\*
		&= h_\theta'(x)v.
		\label{eq:explicit_chain}
	\end{align}
	\begin{remark}\label{rem:chain_rule}
		The dual pass multiplies $v$ by the factors of the product in~\eqref{eq:explicit_chain} from right to left, which is the order of the forward sweep. For all $i\in\{1,\ldots,d-1\}$, it uses $y_i$, whose real part is $a_i$, only in $z_i=\tilde\sigma_i(y_i)$, and thus neither $a_i$ nor $\sigma_i'(a_i)$ is needed after layer $i$. In contrast,~\eqref{eq:backprop} implies that $\delta_{d-1}^\rmT=W_d$ and, for all $i\in\{1,\ldots,d-1\}$, $\delta_{i-1}^\rmT=\delta_i^\rmT\sigma_i'(a_i)W_i$. The backward sweep thus evaluates this product without the factor $v$ from left to right, returns $\delta_0^\rmT=h_\theta'(x)$, and uses the Jacobians in the order $i=d-1,\ldots,1$.\exmark
	\end{remark}
	
	Fig.~\ref{fig:architecture_schematic} shows reverse mode and forward mode with dual numbers and the memory that each mode holds for one evaluation at a state $x$, and Section~\ref{sec:examples} measures the dependence of this memory on the depth.
	
	\begin{figure*}[tb]
		\centering
		\begin{tikzpicture}[
			>={Stealth[length=4.2pt,width=3.4pt]},
			font=\footnotesize,
			blk/.style={rectangle, rounded corners=2.5pt, line width=0.7pt, align=center, text width=3.4cm, minimum width=3.6cm, minimum height=2.5em, inner sep=2.5pt},
			blkrev/.style={blk, draw=mplred!65!black, fill=mplred!5!white},
			blkdual/.style={blk, draw=mplblue!80!black, fill=mplblue!6!white},
			cell/.style={rectangle, rounded corners=1.5pt, line width=0.6pt, align=center, text width=2.9cm, minimum width=3.05cm, minimum height=1.7em, inner sep=1.5pt, fill=white},
			cellrev/.style={cell, draw=mplred!65!black},
			celldual/.style={cell, draw=mplblue!80!black},
			memrev/.style={rectangle, rounded corners=3pt, dashed, line width=0.6pt, draw=mplred!65!black, fill=mplred!3!white, inner sep=5pt},
			memdual/.style={rectangle, rounded corners=3pt, dashed, line width=0.6pt, draw=mplblue!80!black, fill=mplblue!3!white, inner sep=5pt},
			flow/.style={->, line width=0.7pt, draw=black!75},
			arev/.style={->, line width=0.7pt, draw=mplred!65!black},
			adual/.style={<->, line width=0.7pt, draw=mplblue!80!black},
			lbl/.style={font=\scriptsize, inner sep=1.5pt, align=center},
			hdr/.style={font=\small\bfseries},
			sub/.style={font=\scriptsize, text=black!75}
			]
			\node[hdr] at (-4.6,5.3) {Reverse mode};
			\node[sub] at (-4.6,4.92) {One forward sweep and one backward sweep};
			\node[hdr] at (4.6,5.3) {Forward mode with dual numbers};
			\node[sub] at (4.6,4.92) {One forward sweep per direction $v$};
			\draw[dashed, black!35, line width=0.6pt] (0,5.5) -- (0,-0.75);
			\node[blkrev] (in)  at (-2.1,4.15) {Input $x$};
			\node[blkrev] (fwd) at (-2.1,2.85) {Forward sweep\\ $a_i=\ell_i(\hat a_{i-1})$, $\hat a_i=\sigma_i(a_i)$};
			\node[blkrev] (bwd) at (-2.1,1.25) {Backward sweep\\ $\delta_{i-1}=W_i^\rmT\sigma_i'(a_i)\,\delta_i$};
			\node[blkrev] (out) at (-2.1,-0.15) {Outputs $h_\theta(x)$ and $h_\theta'(x)$\\ Products $h_\theta'(x)v$};
			\draw[flow] (in) -- (fwd);
			\draw[flow] (fwd) -- (bwd);
			\draw[flow] (bwd) -- (out);
			\node[font=\footnotesize\bfseries, text=mplred!65!black] (ctitle) at (-6.9,3.78) {Cache};
			\node[cellrev] (c1) at (-6.9,3.2) {Layer $1$: $n_1$ floats};
			\node[cellrev] (c2) at (-6.9,2.55) {Layer $2$: $n_2$ floats};
			\node[text=mplred!65!black, inner sep=0pt] (c3) at (-6.9,2.03) {$\vdots$};
			\node[cellrev] (c4) at (-6.9,1.5) {Layer $d-1$: $n_{d-1}$ floats};
			\begin{scope}[on background layer]
				\node[memrev, fit=(ctitle) (c1) (c4)] (cache) {};
			\end{scope}
			\draw[arev] (fwd.west) -- (fwd.west -| cache.east) node[lbl, midway, above, text=mplred!65!black] {writes};
			\draw[arev] (bwd.west -| cache.east) -- (bwd.west) node[lbl, midway, below, text=mplred!65!black] {reads};
			\node[blkdual] (din) at (2.1,4.15) {Input $z_0=x+v\eps$};
			\node[blkdual] (l1)  at (2.1,3.05) {Layer $1$: $z_1=\tilde\sigma_1(\tilde\ell_1(z_0))$};
			\node[blkdual] (l2)  at (2.1,1.95) {Layer $2$: $z_2=\tilde\sigma_2(\tilde\ell_2(z_1))$};
			\node[blkdual] (lo)  at (2.1,-0.15) {Layer $d$: $z_d=\tilde\ell_d(z_{d-1})$\\ $=h_\theta(x)+h_\theta'(x)v\,\eps$};
			\coordinate (ldc) at ($(l2.south)!0.5!(lo.north)$);
			\node[inner sep=0pt, minimum width=0.3cm, minimum height=0.34cm] (ld) at (ldc) {};
			\foreach \dy in {-0.1,0,0.1} {\fill[mplblue!80!black] ($(ldc)+(0,\dy)$) circle[radius=0.55pt];}
			\draw[flow] (din) -- (l1);
			\draw[flow] (l1) -- (l2);
			\draw[flow] (l2.south) -- (ld.north);
			\draw[flow] (ld.south) -- (lo.north);
			\node[font=\footnotesize\bfseries, text=mplblue!80!black] (btitle) at (6.9,3.15) {Buffer of $w_\eps$ floats};
			\node[celldual] (b1) at (6.9,2.57) {$z_{i-1}$: $2n_{i-1}$ floats};
			\node[celldual] (b2) at (6.9,1.92) {$z_i$: $2n_i$ floats};
			\begin{scope}[on background layer]
				\node[memdual, fit=(btitle) (b1) (b2)] (buf) {};
			\end{scope}
			\draw[adual] (l1.east) -- (l1.east -| buf.west) node[lbl, midway, above, text=mplblue!80!black] {reads $z_0$\\ writes $z_1$};
			\draw[adual] (l2.east) -- (l2.east -| buf.west) node[lbl, midway, below, text=mplblue!80!black] {reads $z_1$\\ writes $z_2$};
		\end{tikzpicture}
		\caption{Evaluation of $h_\theta$ and of a derivative by reverse mode and by forward mode with dual numbers. In reverse mode, the forward sweep returns $h_\theta(x)$ and writes a cache for each hidden layer, and the backward sweep reads the caches in the order $d-1,\ldots,1$ and returns $h_\theta'(x)$, whose products with $f(x)$ and $G(x)$ are the Lie derivatives. If the activations are stored, then the cache comprises $\sum_{i=1}^{d-1}n_i$ floats (Proposition~\ref{prop:memory}\ref{prop2_p1}). In forward mode, for all $i\in\{1,\ldots,d\}$, the evaluation of layer $i$ reads $z_{i-1}$ and writes $z_i$ in a buffer of $w_\eps$ floats (Proposition~\ref{prop:memory}\ref{prop2_p2}). One sweep returns $h_\theta(x)$ and $h_\theta'(x)v$, and the $m+1$ Lie derivatives of the constraint require $m+1$ sweeps.}
		\label{fig:architecture_schematic}
	\end{figure*}
	\section{Second-Order Lie Derivatives}
	\label{sec:higher_order}
	
		If, for all $x\in\SD$, $L_Gh_\theta(x)=0$, then the control input does not appear in~\eqref{eq:hdot} with $h=h_\theta$, and thus the constraint of~\eqref{eq:qp} does not depend on $u$. In this case, $h_\theta$ has relative degree one at no state of $\SD$. This case arises if the input affects the constrained quantity through intermediate states alone, for example, a position constraint on a mechanical system whose input is a force or a torque, or a constraint on a state that the input drives through actuator dynamics. High-order CBFs address this case~\citep{xiao2021high_order_cbf}, and, for relative degree two, the corresponding constraint contains the second-order Lie derivatives $L_f^2 h_\theta$ and $L_GL_f h_\theta$. The hyper-dual numbers of~\citet{fike2011optimization} extend the dual algebra to these quantities. This section parallels Sections~\ref{sec:dual_algebra} and~\ref{sec:dual_nn}, and each definition and result is stated as the counterpart of its first-order version.
	
	\subsection{Hyper-Dual Number Algebra}
	
	\begin{definition}\label{def:hyper}
		The set of \emph{hyper-dual numbers} $\BBD_2$ is the four-dimensional real vector space with basis $\{1,\eps_1,\eps_2,\eps_{12}\}$, whose elements are written $z=a+b_1\eps_1+b_2\eps_2+b_{12}\eps_{12}$, where $a\in\BBR$ is the \emph{real part} and $b_1,b_2,b_{12}\in\BBR$ are the \emph{dual parts} of $z$. Addition is defined componentwise, and multiplication is defined as the bilinear and commutative operation for which $1$ is the identity and
		\begin{equation}
			\eps_1^2=\eps_2^2=0,\ \ \eps_1\eps_2=\eps_{12},\ \ \eps_1\eps_{12}=\eps_2\eps_{12}=\eps_{12}^2=0.
			\label{eq:hyper_table}
		\end{equation}
	\end{definition}
	
	Since $1$, $\eps_1$, $\eps_2$, and $\eps_{12}$ are linearly independent, the real part and the dual parts of each $z\in\BBD_2$ are unique. The real part operator $\Real\colon\BBD_2\to\BBR$ and the dual part operators $\Dual_1,\Dual_2,\Dual_{12}\colon\BBD_2\to\BBR$ are thus defined, for all $z=a+b_1\eps_1+b_2\eps_2+b_{12}\eps_{12}\in\BBD_2$, by $\Real(z)\triangleq a$, $\Dual_1(z)\triangleq b_1$, $\Dual_2(z)\triangleq b_2$, and $\Dual_{12}(z)\triangleq b_{12}$, which are the real, $\eps_1$, $\eps_2$, and $\eps_{12}$ components of $z$, respectively, and we identify $a\in\BBR$ with $a+0\eps_1+0\eps_2+0\eps_{12}$. Since $\eps_{12}$ is a basis element, it is not contained in the span of $1$, $\eps_1$, and $\eps_2$. Since $\BBD_2$ is isomorphic to the quotient of the ring $\BBR[\eps_1,\eps_2]$ of real polynomials in two variables by the ideal generated by $\eps_1^2$ and $\eps_2^2$, in which the class of $\eps_1\eps_2$ corresponds to $\eps_{12}$, the triple $(\BBD_2,+,\cdot\,)$ is a commutative ring. As in Fact~\ref{fact:ring}, this ring is not a field, because $\eps_1\neq0$ and $\eps_1^2=0$.
	
	Definition~\ref{def:hyper_ext} is the counterpart of Definition~\ref{def:dual_ext} for $\BBD_2$.
	
	\begin{definition}\label{def:hyper_ext}
		Let $I\subseteq\BBR$ be open, and let $\varphi\colon I\to\BBR$ be twice differentiable at $a\in I$. For all $b_1,b_2,b_{12}\in\BBR$, the \emph{hyper-dual extension} of $\varphi$ at $z=a+b_1\eps_1+b_2\eps_2+b_{12}\eps_{12}\in\BBD_2$ is
		\begin{align}
			\tilde\varphi(z) &\triangleq \varphi(a)+\varphi'(a)b_1\eps_1+\varphi'(a)b_2\eps_2\nn\\*
			&\quad+\bigl(\varphi''(a)b_1b_2+\varphi'(a)b_{12}\bigr)\eps_{12},
			\label{eq:hyper_scalar}
		\end{align}
		and thus its real, $\eps_1$, $\eps_2$, and $\eps_{12}$ components are $\varphi(a)$, $\varphi'(a)b_1$, $\varphi'(a)b_2$, and $\varphi''(a)b_1b_2+\varphi'(a)b_{12}$, respectively.
	\end{definition}
	
	For $b_2=b_{12}=0$, \eqref{eq:hyper_scalar} is~\eqref{eq:dual_RD} with $b_1$ and $\eps_1$ in place of $b$ and $\eps$, respectively, and, in this section, the tilde therefore denotes the hyper-dual extension. Proposition~\ref{thm:hyper_exact_diff} \citep{fike2011optimization} states the relation of the hyper-dual extension to the arithmetic of $\BBD_2$, namely, the sum, product, and chain rules for first and second derivatives, and Appendix~\ref{app:thm_hyper_exact_diff} contains the proof.
	
	\begin{proposition}\label{thm:hyper_exact_diff}
		Let $I\subseteq\BBR$ be open, let $\varphi,\psi\colon I\to\BBR$ be twice differentiable at $a\in I$, let $b_1,b_2,b_{12}\in\BBR$, and let $z\triangleq a+b_1\eps_1+b_2\eps_2+b_{12}\eps_{12}\in\BBD_2$. Then, the following statements hold:
		\begin{enumerate}[label={\it{\roman*}})]
			\item\label{hp_p1} $\widetilde{(\varphi+\psi)}(z)=\tilde\varphi(z)+\tilde\psi(z)$.
			\item\label{hp_p2} $\widetilde{(\varphi\psi)}(z)=\tilde\varphi(z)\,\tilde\psi(z)$.
			\item\label{hp_p3} If $J\subseteq\BBR$ is open, $\varphi(I)\subseteq J$, and $\chi\colon J\to\BBR$ is twice differentiable at $\varphi(a)$, then $\widetilde{(\chi\circ\varphi)}(z)=\tilde\chi(\tilde\varphi(z))$.
			\item\label{hp_p4} If $r\in\BBN$ and $\gamma_0,\ldots,\gamma_r\in\BBR$ satisfy, for all $t\in I$, $\varphi(t)=\sum_{i=0}^{r}\gamma_it^i$, then $\tilde\varphi(z)=\sum_{i=0}^{r}\gamma_iz^i$, where $z^0\triangleq1$ and the powers are computed using~\eqref{eq:hyper_table}.
		\end{enumerate}
	\end{proposition}
	
	Sums and products of hyper-dual vectors $x+v_1\eps_1+v_2\eps_2+v_{12}\eps_{12}\in\BBD_2^n$, where $x,v_1,v_2,v_{12}\in\BBR^n$, and of hyper-dual matrices are defined by the matrix operations with the arithmetic of Definition~\ref{def:hyper} for the entries, and $\Real$, $\Dual_1$, $\Dual_2$, and $\Dual_{12}$ act entrywise.
	
	\begin{definition}\label{def:hyper_vec}
		Let $\SD\subseteq\BBR^n$ be open, and let $\varphi\colon\SD\to\BBR^p$ be twice differentiable at $x\in\SD$. For all $v_1,v_2,v_{12}\in\BBR^n$, the \emph{hyper-dual extension} of $\varphi$ at $z=x+v_1\eps_1+v_2\eps_2+v_{12}\eps_{12}\in\BBD_2^n$ is
		\begin{align}
			\tilde\varphi(z) &\triangleq \varphi(x)+\varphi'(x)v_1\eps_1+\varphi'(x)v_2\eps_2 \nn\\*
			&\quad+\bigl(\varphi''(x)[v_1,v_2]+\varphi'(x)v_{12}\bigr)\eps_{12}\in\BBD_2^p,
			\label{eq:hyper_vec}
		\end{align}
		where $\varphi''(x)[v_1,v_2]$ is defined in Section~\ref{sec:notation}.
	\end{definition}
	
	For $n=p=1$, Definition~\ref{def:hyper_vec} coincides with Definition~\ref{def:hyper_ext}. Corollary~\ref{cor:hyper_vec}, which Appendix~\ref{app:thm_hyper_exact_diff} also proves, extends Proposition~\ref{thm:hyper_exact_diff}\ref{hp_p3} to vector functions.
	\begin{corollary}\label{cor:hyper_vec}
		Let $\SD\subseteq\BBR^n$ and $\mathcal{E}\subseteq\BBR^p$ be open, let $\varphi\colon\SD\to\mathcal{E}$ be twice differentiable at $x\in\SD$, and let $\psi\colon\mathcal{E}\to\BBR^q$ be twice differentiable at $\varphi(x)$. Then, for all $v_1,v_2,v_{12}\in\BBR^n$, $z\triangleq x+v_1\eps_1+v_2\eps_2+v_{12}\eps_{12}$ satisfies
		\begin{equation}
			\tilde\psi\bigl(\tilde\varphi(z)\bigr)=\widetilde{(\psi\circ\varphi)}(z).
			\label{eq:hyper_chain}
		\end{equation}
	\end{corollary}
	
	\subsection{Hyper-Dual Extension of Network Layers}
	
	Lemmas~\ref{lem:hyper_affine} and~\ref{lem:hyper_activation} are the counterparts of Lemmas~\ref{lem:dual_affine} and~\ref{lem:dual_activation} for $\BBD_2$, and Appendix~\ref{app:thm_hyper_exact_diff} contains their proofs.
	
	\begin{lemma}\label{lem:hyper_affine}
		Let $i\in\{1,\ldots,d\}$. Then, for all $a,c_1,c_2,c_{12}\in\BBR^{n_{i-1}}$, the hyper-dual extension of the affine layer $\ell_i$ at $z\triangleq a+c_1\eps_1+c_2\eps_2+c_{12}\eps_{12}$ satisfies
		\begin{align}
			\tilde\ell_i(z) &= (W_i a+b_i)+W_i c_1\eps_1+W_i c_2\eps_2\nn\\*
			&\quad+W_i c_{12}\eps_{12}.
			\label{eq:hyper_affine}
		\end{align}
	\end{lemma}
	
	\begin{lemma}\label{lem:hyper_activation}
		Let $i\in\{1,\ldots,d-1\}$, let $a,c_1,c_2,c_{12}\in\BBR^{n_i}$, and assume that, for all $j\in\{1,\ldots,n_i\}$, $\bar\sigma_i$ is twice differentiable at $a_{(j)}$. Then, the hyper-dual extension of the activation layer $\sigma_i$ at $z\triangleq a+c_1\eps_1+c_2\eps_2+c_{12}\eps_{12}$ satisfies
		\begin{align}
			\tilde\sigma_i(z) &= \sigma_i(a)+\sigma_i'(a)c_1\eps_1+\sigma_i'(a)c_2\eps_2\nn\\*
			&\quad+\bigl(\sigma_i''(a)[c_1,c_2]+\sigma_i'(a)c_{12}\bigr)\eps_{12},
			\label{eq:hyper_activation}
		\end{align}
		where $\sigma_i'(a)$ is the diagonal matrix~\eqref{eq:sigma_prime} and, for all $j\in\{1,\ldots,n_i\}$, $(\sigma_i''(a)[c_1,c_2])_{(j)}=\bar\sigma_i''(a_{(j)})c_{1(j)}c_{2(j)}$. Furthermore, for all $j\in\{1,\ldots,n_i\}$, the $j$th component of~\eqref{eq:hyper_activation} is the scalar extension~\eqref{eq:hyper_scalar} of $\bar\sigma_i$ at $z_{(j)}$.
	\end{lemma}
	
	Lemma~\ref{lem:hyper_activation} requires, for all $i\in\{1,\ldots,d-1\}$, the scalar activation $\bar\sigma_i$ to be twice differentiable at the components of the preactivation of layer $i$. For a network with the ReLU activation, $h_\theta$ is piecewise affine, and, at a state of an activation boundary, $L_f h_\theta$ is discontinuous if the gradients of the two adjacent affine pieces have different inner products with $f$. Since the high-order CBF condition~\citep{xiao2021high_order_cbf} contains second-order Lie derivatives, it does not apply to ReLU networks. The second-order extension therefore uses twice continuously differentiable activations, such as the hyperbolic tangent, the sigmoid, and the softplus, and the compiler emits no second-order code for ReLU networks. For all $i\in\{1,\ldots,d-1\}$ such that $\bar\sigma_i$ is the hyperbolic tangent, the sigmoid, or the softplus, the hyper-dual layer follows from~\eqref{eq:hyper_activation} with, for all $t\in\BBR$, $\bar\sigma_i''(t)=-2\bar\sigma_i(t)\bar\sigma_i'(t)$, $\bar\sigma_i''(t)=\bar\sigma_i'(t)(1-2\bar\sigma_i(t))$, and $\bar\sigma_i''(t)=\bar\sigma_i'(t)(1-\bar\sigma_i'(t))$, respectively. These expressions require no further evaluation of the exponential or the logarithm.
	
	\subsection{Main Result: Exact Second-Order Lie Derivatives}
	
	The Lie derivatives $L_f^2 h_\theta(x)$ and $L_GL_f h_\theta(x)$ depend on the Jacobian $h_\theta'(x)$ and the Hessian $h_\theta''(x)$. A \emph{hyper-dual seed} is a hyper-dual vector of the form $x+f(x)\eps_1+\xi(x)\eps_2+f'(x)\xi(x)\eps_{12}\in\BBD_2^n$, where $x\in\SD$ and $\xi\colon\SD\to\BBR^n$ is a vector field. The following theorem states that the evaluation of the hyper-dual-extended network at a hyper-dual seed returns the Lie derivatives along $f$ and $\xi$ and the second-order Lie derivative $L_\xi L_fh_\theta(x)$ in the dual parts of the output.
	
	\begin{theorem}\label{thm:hyper_lie}
		Consider the neural network $h_\theta$ defined in~\eqref{eq:nn_def}, and assume that~\ref{assum:dynamics} holds and that, for all $i\in\{1,\ldots,d-1\}$, $\bar\sigma_i$ is twice continuously differentiable. Let $x\in\SD$, define the \emph{hyper-dual-extended network} $\tilde h_\theta$ by~\eqref{eq:h_tilde_def} with the hyper-dual extensions of the layers, and, for all $\xi\colon\SD\to\BBR^n$, define
		\begin{equation}
			x_{f,\xi}\triangleq x+f(x)\eps_1+\xi(x)\eps_2+f'(x)\xi(x)\eps_{12}\in\BBD_2^n.
			\label{eq:seed_generic}
		\end{equation}
		Then, $h_\theta$ is twice continuously differentiable on $\SD$, and, for all $\xi\colon\SD\to\BBR^n$,
		\begin{align}
			\tilde h_\theta(x_{f,\xi}) &= h_\theta(x)+L_f h_\theta(x)\eps_1+L_\xi h_\theta(x)\eps_2\nn\\*
			&\quad+L_\xi L_f h_\theta(x)\eps_{12}.
			\label{eq:hyper_main_result}
		\end{align}
		In particular, $\Dual_{12}(\tilde h_\theta(x_{f,f}))=L_f^2 h_\theta(x)$, and, for all $j\in\{1,\ldots,m\}$, $x_{f,G,j}\triangleq x_{f,\col_j(G)}$ satisfies
		\begin{equation}
			\Dual_{12}\bigl(\tilde h_\theta(x_{f,G,j})\bigr)=\col_j\bigl(L_GL_f h_\theta(x)\bigr).\label{eq:extract_lglf}
		\end{equation}
	\end{theorem}
	
	\begin{proof}
		Since, for all $i\in\{1,\ldots,d-1\}$, $\bar\sigma_i$ is twice continuously differentiable and, for all $i\in\{1,\ldots,d\}$, $\ell_i$ is affine, the chain rule implies that $h_\theta$ is twice continuously differentiable on $\SD$. The induction in the proof of Theorem~\ref{thm:main}, with Corollary~\ref{cor:hyper_vec} in place of Corollary~\ref{cor:vec_ext}, implies that, for all $z\in\BBD_2^n$ such that $\Real(z)=x$, $\tilde h_\theta(z)$ is the hyper-dual extension of $h_\theta$ at $z$. Let $\xi\colon\SD\to\BBR^n$. Since $h_\theta'(x)f(x)=L_fh_\theta(x)$ and $h_\theta'(x)\xi(x)=L_\xi h_\theta(x)$, Definition~\ref{def:hyper_vec} at the seed~\eqref{eq:seed_generic} yields
		\begin{align}
			\tilde h_\theta(x_{f,\xi}) &= h_\theta(x)+L_f h_\theta(x)\eps_1+L_\xi h_\theta(x)\eps_2 \nn\\*
			&\quad+\Bigl(h_\theta''(x)\bigl[f(x),\xi(x)\bigr]\nn\\*
			&\quad\quad+h_\theta'(x)f'(x)\xi(x)\Bigr)\eps_{12}.
			\label{eq:iterated_hyper_chain}
		\end{align}
		Since~\ref{assum:dynamics} holds, the function $L_fh_\theta$, whose value at $x$ is $h_\theta'(x)f(x)$, is continuously differentiable on $\SD$, and the product rule implies that its Jacobian at $x$ is $f(x)^\rmT h_\theta''(x)+h_\theta'(x)f'(x)$. The definition of the Lie derivative thus yields
		\begin{align}
			&L_\xi L_fh_\theta(x)=\bigl(f(x)^\rmT h_\theta''(x)+h_\theta'(x)f'(x)\bigr)\xi(x)\nn\\*
			&\qquad=h_\theta''(x)\bigl[f(x),\xi(x)\bigr]+h_\theta'(x)f'(x)\xi(x).
			\label{eq:l2f_identity}
		\end{align}
		Substituting~\eqref{eq:l2f_identity} into~\eqref{eq:iterated_hyper_chain} yields~\eqref{eq:hyper_main_result}. Setting $\xi=f$ in~\eqref{eq:hyper_main_result} yields $\Dual_{12}(\tilde h_\theta(x_{f,f}))=L_f^2h_\theta(x)$. For all $j\in\{1,\ldots,m\}$, setting $\xi=\col_j(G)$ in~\eqref{eq:hyper_main_result} yields~\eqref{eq:extract_lglf}, because $L_{\col_j(G)}L_fh_\theta(x)=\col_j(L_GL_fh_\theta(x))$.
	\end{proof}

	Theorem~\ref{thm:hyper_lie} with $\xi=f$ implies that one evaluation of $\tilde h_\theta$ at $x_{f,f}$, called a \emph{hyper-dual pass}, returns $h_\theta(x)$, $L_f h_\theta(x)$, and $L_f^2 h_\theta(x)$, which are the drift quantities of the high-order constraint of Experiment~\ref{ex:pendulum}. Since $L_\xi L_fh_\theta(x)$ depends on $\xi$ through $\xi(x)$ alone, the seed~\eqref{eq:seed_generic} contains $f'(x)$ and no derivative of $\xi$, and thus $L_GL_fh_\theta(x)$ requires $f'(x)G(x)$ and no derivative of $G$.
	
	\begin{remark}\label{rem:jet}
		The four components of $\BBD_2$ are a property of this construction and not a memory requirement of second-order forward differentiation. The truncated algebra of order two with basis $\{1,\eta_1,\eta_2\}$ and products $\eta_1^2=2\eta_2$ and $\eta_1\eta_2=\eta_2^2=0$ has three components \citep{penunuri2020dual}. For all $\varphi\colon\SD\to\BBR^p$ that are twice differentiable at $x\in\SD$ and all $v_1,v_2\in\BBR^n$, the extension of $\varphi$ at $z=x+v_1\eta_1+v_2\eta_2$ is
		\begin{align}
			\check\varphi(z)&\triangleq\varphi(x)+\varphi'(x)v_1\eta_1\nn\\*
			&\quad+\bigl(\varphi''(x)[v_1,v_1]+\varphi'(x)v_2\bigr)\eta_2.
			\label{eq:jet_ext}
		\end{align}
		As in the proof of Corollary~\ref{cor:hyper_vec}, with $[v_1,v_1]$ in place of $[v_1,v_2]$, the extension of a composition is the composition of the extensions, and thus a layer-by-layer pass returns the extension $\check h_\theta$ of $h_\theta$. It follows from~\eqref{eq:jet_ext} and from~\eqref{eq:l2f_identity} with $\xi=f$ that the components of $\check h_\theta$ at the seed $x+f(x)\eta_1+f'(x)f(x)\eta_2$ are $h_\theta(x)$, $L_fh_\theta(x)$, and $L_f^2h_\theta(x)$. Since the $\eta_2$ component of~\eqref{eq:jet_ext} contains the second derivative along the single direction $v_1$, the mixed term of $L_GL_fh_\theta(x)$ requires two passes per input. For all $j\in\{1,\ldots,m\}$, let $\xi_j\triangleq\col_j(G)$, and let $q_j^+$ and $q_j^-$ denote the $\eta_2$ components of $\check h_\theta$ at the seeds $x+(f(x)+\xi_j(x))\eta_1+2f'(x)\xi_j(x)\eta_2$ and $x+(f(x)-\xi_j(x))\eta_1-2f'(x)\xi_j(x)\eta_2$, respectively. For all $v,w\in\BBR^n$, the symmetry of $\varphi''(x)$ implies the polarization identity $\varphi''(x)[v+w,v+w]-\varphi''(x)[v-w,v-w]=4\varphi''(x)[v,w]$, and the factor $2$ in the $\eta_2$ components of the two seeds yields the same factor $4$ in the difference of the terms $\varphi'(x)v_2$. Equation~\eqref{eq:jet_ext} with $\varphi=h_\theta$ thus implies that, for all $j\in\{1,\ldots,m\}$,
		\begin{equation}
			\col_j\bigl(L_GL_fh_\theta(x)\bigr)=L_{\xi_j}L_fh_\theta(x)=\tfrac{1}{4}\bigl(q_j^+-q_j^-\bigr).
			\label{eq:polarization}
		\end{equation}
		The relative-degree-two constraint thus requires $m+1$ passes of four components with $\BBD_2$ and $2m+1$ passes of three components with the truncated algebra. The workspace per pass is $2w_\eps$ and $3w_\eps/2$ floats, and the numbers of propagated components are $4(m+1)$ and $3(2m+1)$, respectively. For $m=1$, the truncated algebra uses $25\%$ less workspace and propagates $9$ components instead of $8$, and the difference in~\eqref{eq:polarization} is subject to cancellation in floating-point arithmetic. The compiler emits the hyper-dual form, which, for all $m\ge1$, requires fewer passes and, since $3(2m+1)-4(m+1)=2m-1>0$, fewer propagated components, and Section~\ref{sec:examples} measures both forms. Relative degrees greater than two can use families of univariate Taylor propagations \citep{griewank2000evaluating}, which this paper does not develop.\exmark
	\end{remark}

	\section{Embedded Validation and Numerical Profiling}
	\label{sec:examples}
	
	This section evaluates the compiler on three nonlinear control systems and on networks of increasing depth. All networks are trained in PyTorch by least-squares regression onto an analytic target barrier, with the Adam optimizer, $5{,}000$ iterations at the learning rate $10^{-3}$, batches of $1{,}024$ states drawn uniformly from the sampling box of the experiment, and the random seed $0$, exported to the ONNX format, and translated by the compiler into a C++ header. The network of Experiment~\ref{ex:bicycle} instead uses $20{,}000$ iterations, the last $6{,}000$ at the learning rate $2\times10^{-4}$. A bar marks a physical quantity whose letter has another meaning in Sections~\ref{sec:notation}--\ref{sec:higher_order}. The trained networks are not verified to satisfy the condition of a CBF on their sampling boxes. The ReLU networks of Experiments~\ref{ex:bicycle} and~\ref{ex:neural} are computational demonstrations, and the softplus barrier of Experiment~\ref{ex:pendulum} is checked numerically on a grid.
	
	\textit{Compiler.} The compiler proceeds in three stages. First, it reads the model and determines the topology $(n_0,\ldots,n_d)$ and the activations. Second, it sizes the buffers from the layer widths. Third, it emits C++ code with one nested loop per layer and with the weights embedded as constant arrays, without statements that allocate memory, without recursion, and without \emph{data-dependent branches} (i.e., conditional statements whose outcome depends on the data) other than the ReLU comparisons. The loop bounds are the layer widths, and the weights are data. The networks of Experiments~\ref{ex:bicycle} and~\ref{ex:neural} have $1{,}249$ and $4{,}417$ parameters, a ratio of $3.5$, and, with the compiler and flags of the next paragraph, the code sizes of their sequential entry points are $1{,}442$ and $1{,}579$~bytes, a ratio of $1.1$ (Table~\ref{tab:measurements}). For the networks with $2$, $4$, and $8$ hidden layers of width $32$ of the depth sweep, these code sizes are $1{,}393$, $2{,}485$, and $4{,}669$~bytes, which equals $301+546(d-1)$~bytes.

	\textit{Platform.} The code is compiled with ESP-IDF~v5.5.4 (i.e., the software development framework of the ESP32-S3), which uses GCC~14.2.0, with the optimization flags \texttt{-O2 -ffast-math -funroll-loops}, which apply to all methods, and runs on an ESP32-S3 of chip revision~v0.2, a dual-core 32-bit Xtensa LX7 processor with a single-precision floating-point unit and $512$~KB of internal static random-access memory. Double-precision arithmetic is emulated in software and is used for the reference alone. The processor clock is fixed at $240$~MHz, and power management, Wi-Fi, and Bluetooth are disabled. The weights reside in \emph{flash}, that is, the nonvolatile memory of the microcontroller, and are read through the \emph{data cache} of the processor, which is a memory that holds copies of recently read data and is read in less time than the flash. The benchmark runs as one task of the FreeRTOS operating system on one of the two cores. Since \texttt{-ffast-math} permits the compiler to reassociate floating-point operations and to assume finite operands, the reported errors apply to this build. The emitted code compiles without this flag for deployments that require strict IEEE~754 semantics, and such a build is not measured.
	
	\textit{Timing protocol.} The cycle counter \texttt{CCOUNT} of the processor, which holds the number of elapsed clock cycles, is read before and after each call, and the difference is converted to time at $240$ cycles per microsecond. After $16$ warm-up calls, which are not timed, each method is called through a function with the same arguments that is not \emph{inlined} (i.e., not substituted into the calling code) at the same $1{,}000$ states, which a pseudorandom generator with a fixed seed draws uniformly from the sampling box. We report two timed windows. The assembly window covers lines~1--6 of Algorithm~\ref{alg:dual_cbf}. The step window adds the evaluation of $f(x)$ and $G(x)$, the evaluation of $f'(x)f(x)$ and $f'(x)G(x)$ in Experiment~\ref{ex:pendulum}, and lines~7--11, which solve the QP or return the fallback input. In the experiments, $\SU$ is a box. Let $x\in\SD$, and define $\omega\triangleq L_Gh_\theta(x)^\rmT\in\BBR^m$ and $\beta\triangleq-L_fh_\theta(x)-\alpha(h_\theta(x))$, and thus the constraint of~\eqref{eq:qp} is $\omega^\rmT u\ge\beta$. In Experiment~\ref{ex:pendulum}, $\omega$ and $\beta$ are the corresponding coefficients of~\eqref{eq:hocbf}. For all $\lambda\ge0$, let $u(\lambda)$ denote the projection of $u_\mathrm{n}(x)+\lambda\omega$ onto $\SU$. The function $\lambda\mapsto\omega^\rmT u(\lambda)$ is continuous, nondecreasing, and piecewise affine, and its breakpoints are the values of $\lambda$ at which a component of $u_\mathrm{n}(x)+\lambda\omega$ reaches a bound of $\SU$, which number at most $2m$. If there exists $\lambda\ge0$ such that $\omega^\rmT u(\lambda)\ge\beta$, then the Karush--Kuhn--Tucker conditions of~\eqref{eq:qp} imply that its minimizer is $u(\lambda_*)$, where $\lambda_*\triangleq\min\{\lambda\ge0\colon\omega^\rmT u(\lambda)\ge\beta\}$, and a search over the breakpoints computes $\lambda_*$. Otherwise, the search returns $\lim_{\lambda\to\infty}u(\lambda)$. Since the projection onto a box is componentwise, this limit is the fallback input of line~10. A build of the same code on a desktop computer, called the \emph{host}, at the same states shows that the fallback occurs at $0$, $8$, and $285$ of the $1{,}000$ timing states of Experiments~\ref{ex:bicycle}, \ref{ex:neural}, and~\ref{ex:pendulum} and at $8$, $8$, and $4$ states for the networks of Table~\ref{tab:depth}, and the step statistics are taken over all $1{,}000$ states. An \emph{interrupt} is an event that suspends the running code. Each window is measured with interrupts enabled, in which case the periodic $100$~Hz interrupt of the operating system can fall inside the window, and with interrupts disabled, which is called \emph{masked}, on the core that runs the benchmark. We report the median and the maximum observed execution time. The maximum is a sample maximum and not a bound on the WCET, and no static WCET analysis (i.e., an analysis of the program without its execution) is performed. Each method is measured in one run, and thus the run-to-run variation is not estimated and the ratios of medians between $0.93$ and $1.07$ do not establish an ordering of the corresponding methods.
	
	\textit{Methods.} All methods read the same weight arrays, and the principal baselines are generated ahead of time for the same topology as the compiled header. For relative degree one, we compare the sequential and the shared-primal entry points of the compiled header with a reverse-mode implementation whose forward sweep stores the postactivations of the hidden layers in static arrays and whose backward sweep uses two static adjoint buffers (Remark~\ref{rem:static_reverse}). As a secondary comparison, we include a reverse-mode implementation that allocates its cache from the heap at each call, as general-purpose AD frameworks do. For relative degree two, we compare the hyper-dual header, the three-component algebra of Remark~\ref{rem:jet}, and an exact second-order reverse-mode implementation with statically allocated caches. This forward-over-reverse implementation executes a forward sweep that also propagates the derivatives along $f$, followed by one backward sweep, and it returns $h_\theta'(x)$ and $f(x)^\rmT h_\theta''(x)$, from which the Jacobian $f(x)^\rmT h_\theta''(x)+h_\theta'(x)f'(x)$ of $L_fh_\theta$ at $x$ yields $L_f^2h_\theta(x)$ and $L_GL_fh_\theta(x)$. As a secondary comparison, we include an implementation that forms the Hessian--vector products by one-sided finite differences of reverse-mode gradients with step $10^{-3}$ and heap-allocated caches.
	
	\textit{Memory.} Memory is reported in four parts. Static storage is the size of the static arrays, which is read from the compiled program. Stack usage is measured by filling the stack below the caller with a pattern, running the method with interrupts masked, and locating the deepest overwritten byte. Heap usage is measured with the allocation hooks of ESP-IDF, which count the allocations made by the method and the library routines while the method runs. Code size is the size in flash of the functions of each method. The working memory of a method is the sum of its static storage, its stack usage, and its peak heap usage. The parameters occupy $4n_\theta$ bytes of flash for all methods, which are $4{,}996$, $17{,}668$, and $4{,}740$~bytes in Experiments~\ref{ex:bicycle}, \ref{ex:neural}, and~\ref{ex:pendulum}, and $4{,}740$, $13{,}188$, and $30{,}084$~bytes for the networks with $2$, $4$, and $8$ hidden layers of Table~\ref{tab:depth}. For each experiment, the binary file written to the flash, called the \emph{firmware image}, contains all methods, the shared library routines, the measurement code, and the ESP-IDF system software, and the paper does not attribute the size of this image to individual methods.
	
	\begin{table*}[t]
		\centering
		\caption{ESP32-S3 measurements. Times are in $\mu$s over $1{,}000$ evaluations at seeded states. Med.\ and max.\ are the median and the maximum observed time with interrupts enabled, and mask.\ is the maximum observed time with interrupts masked. The assembly window covers lines 1--6 of Algorithm~\ref{alg:dual_cbf}, and the step window covers, in addition, the evaluation of the dynamics and the QP of lines 7--11. Memory is in bytes, excludes the parameters in flash, and consists of static buffers, stack, peak heap, and the total, which is the working memory. Code is the size in flash of the functions of the method. Err.\ is the maximum scaled error against the double-precision reference over the states at least $10^{-4}$ from all ReLU activation boundaries, and the text reports the error at the other states. VdP denotes the Van der Pol oscillator, and FD denotes finite differences.}
		\label{tab:measurements}
		\footnotesize
		\setlength{\tabcolsep}{4.35pt}
		\begin{tabular}{llrrrrrrrrrrr}
			\toprule
			& & \multicolumn{3}{c}{Assembly ($\mu$s)} & \multicolumn{2}{c}{Step ($\mu$s)} & \multicolumn{4}{c}{Memory (B)} & & \\
			\cmidrule(lr){3-5}\cmidrule(lr){6-7}\cmidrule(lr){8-11}
			Experiment & Method & Med. & Max. & Mask. & Med. & Max. & Static & Stack & Heap & Total & Code (B) & Err. \\
			\midrule
			Bicycle & Dual, sequential & 97.00 & 103.19 & 97.30 & 98.18 & 105.48 & 256 & 296 & 0 & 552 & 1{,}442 & $1.2\times10^{-5}$ \\
			& Dual, shared primal & 61.61 & 67.69 & 62.09 & 62.73 & 70.00 & 512 & 600 & 0 & 1{,}112 & 1{,}627 & $1.2\times10^{-5}$ \\
			& Reverse, static & 47.89 & 53.97 & 48.11 & 49.00 & 56.25 & 512 & 24 & 0 & 536 & 1{,}374 & $7.2\times10^{-6}$ \\
			& Reverse, heap & 143.56 & 150.30 & 144.29 & 144.65 & 152.62 & 0 & 328 & 816 & 1{,}144 & 4{,}538 & $7.2\times10^{-6}$ \\
			VdP & Dual, sequential & 237.29 & 243.57 & 237.88 & 237.46 & 244.20 & 512 & 552 & 0 & 1{,}064 & 1{,}579 & $5.3\times10^{-7}$ \\
			& Dual, shared primal & 181.41 & 187.56 & 182.73 & 181.55 & 188.29 & 768 & 824 & 0 & 1{,}592 & 1{,}401 & $5.3\times10^{-7}$ \\
			& Reverse, static & 169.73 & 175.81 & 170.01 & 169.86 & 176.01 & 1{,}024 & 152 & 0 & 1{,}176 & 1{,}410 & $9.5\times10^{-7}$ \\
			& Reverse, heap & 281.98 & 288.72 & 284.09 & 282.10 & 289.40 & 0 & 328 & 1{,}576 & 1{,}904 & 4{,}258 & $8.9\times10^{-7}$ \\
			Pendulum & Hyper-dual & 510.78 & 520.27 & 515.17 & 511.40 & 521.89 & 512 & 696 & 0 & 1{,}208 & 1{,}733 & $6.0\times10^{-7}$ \\
			& Three-component & 739.40 & 751.22 & 745.80 & 740.00 & 752.01 & 384 & 584 & 0 & 968 & 1{,}739 & $1.2\times10^{-6}$ \\
			& Fwd.-over-rev., static & 272.13 & 280.38 & 274.42 & 272.78 & 281.20 & 1{,}280 & 440 & 0 & 1{,}720 & 2{,}561 & $6.0\times10^{-7}$ \\
			& Reverse, heap, FD & 986.51 & 998.53 & 992.53 & 987.01 & 999.60 & 0 & 408 & 800 & 1{,}208 & 3{,}818 & $6.5\times10^{-4}$ \\
			\bottomrule
		\end{tabular}
	\end{table*}
	
	\textit{Memory layout.} Let $n_\mathrm{c}$ denote the number of floats that an implementation propagates per neuron, namely, $n_\mathrm{c}=2$, $n_\mathrm{c}=m+2$, and $n_\mathrm{c}=4$ for the sequential, shared-primal, and hyper-dual entry points and $n_\mathrm{c}=3$ for the three-component implementation. Each of them holds the state of the preceding layer in a static array of $n_\mathrm{c}\max_{i\in\{0,\ldots,d\}}n_i$ floats and, for all $i\in\{1,\ldots,d\}$, the state of layer $i$ in an automatic array of $n_\mathrm{c}n_i$ floats. Since the automatic array of layer $i$ exists only while the code of layer $i$ runs, the stack holds one of these arrays at a time, and the stack usage does not grow with the depth (Table~\ref{tab:depth}). The static array and the largest automatic array occupy $n_\mathrm{c}(\max_{i\in\{0,\ldots,d\}}n_i+\max_{i\in\{1,\ldots,d\}}n_i)$ floats. This number is at least $n_\mathrm{c}w_\eps/2$, with equality for the networks of this section, in which two layers of maximum width are adjacent. It thus equals $w_\eps$ of Proposition~\ref{prop:memory}\ref{prop2_p2} for the sequential entry point, the workspace $(m+2)\max_{i\in\{1,\ldots,d\}}(n_{i-1}+n_i)$ of Remark~\ref{rem:ad_cost} for the shared-primal entry point, and $2w_\eps$ and $3w_\eps/2$ for the hyper-dual entry point and the three-component implementation (Remark~\ref{rem:jet}). Tables~\ref{tab:measurements} and~\ref{tab:depth} report the static arrays under Static, and the automatic arrays and the \emph{call frames} (i.e., the stack memory that each function call reserves) under Stack. Since all calls of an entry point share its static array, the entry points are not \emph{reentrant}, which means that two calls of the same entry point must not overlap in time.
	
	\textit{Numerical error.} The numerical error is measured against an evaluation of the same network, with the same single-precision parameters, in double-precision arithmetic at $10{,}000$ seeded states per experiment. This reference is cross-checked against \texttt{torch.autograd} in double precision, with which it agrees to within $10^{-13}$. For each method, the reported error is the maximum over the states and over the returned quantities $q$ of $|q-q_\mathrm{r}|/\max\{1,|q_\mathrm{r}|\}$, where $q_\mathrm{r}$ is the reference value. For the ReLU networks, the states at which a preactivation of the reference is within $10^{-4}$ of zero are counted and excluded from this maximum (Remark~\ref{rem:kinks}). The experiments report the number of these states and the error at them, and the maximum over all states is the larger of the two values. Table~\ref{tab:measurements} summarizes the measurements.

	\begin{experiment}[Lane-keeping barrier for a kinematic bicycle]\label{ex:bicycle}
		Consider a vehicle governed by a kinematic bicycle model with wheelbase $\bar\ell_\mathrm{w}>0$~\citep{rajamani2006vehicle}. Let $x=[\bar p_x\ \ \bar p_y\ \ \bar\psi\ \ \bar v]^\rmT\in\BBR^4$ denote the state, which consists of the position $(\bar p_x,\bar p_y)$, the heading $\bar\psi$, and the speed $\bar v$. The physical inputs are the longitudinal acceleration $\bar a$ and the front steering angle $\bar\delta$, and the heading satisfies $\dot{\bar\psi}=(\bar v/\bar\ell_\mathrm{w})\tan\bar\delta$. Since $\tan(\cdot)$ is a bijection from $(-\pi/2,\pi/2)$ onto $\BBR$, we use the input $u\triangleq[\bar a\ \ \tan\bar\delta]^\rmT\in\BBR^2$, for which the model is control-affine without approximation, and the steering angle is $\bar\delta=\arctan u_{(2)}$. The dynamics~\eqref{eq:dynamics} then take the form
		\begin{equation}
			f(x) = \matl \bar v\cos\bar\psi \\ \bar v\sin\bar\psi \\ 0 \\ 0 \matr,\quad G(x) = \matl 0 & 0 \\ 0 & 0 \\ 0 & \bar v/\bar\ell_\mathrm{w} \\ 1 & 0 \matr.
			\label{eq:bicycle_dynamics}
		\end{equation}
		We parameterize $h_\theta$ by a ReLU network with architecture $4$-$32$-$32$-$1$, and thus $d=3$ and $C_\mathrm{f}=2{,}432$ operations. The wheelbase is $\bar\ell_\mathrm{w}=2.5$~m. The target barrier is $h(x)=\bar w^2-(\bar p_y+\tau_\mathrm{h}\bar v\sin\bar\psi)^2$, where $\bar w=3.5$~m is the half-width of the lane and $\bar p_y+\tau_\mathrm{h}\bar v\sin\bar\psi$ is the \emph{look-ahead lateral position} with the look-ahead time $\tau_\mathrm{h}=1$~s, and thus
		\begin{equation}
			\setlength{\arraycolsep}{3pt}L_Gh(x)=-2\tau_\mathrm{h}(\bar p_y+\tau_\mathrm{h}\bar v\sin\bar\psi)\matl\sin\bar\psi & \bar v^2\cos\bar\psi/\bar\ell_\mathrm{w}\matr.
			\label{eq:bicycle_LG}
		\end{equation}
		Thus, for all $x\in\BBR^4$ such that $\bar p_y+\tau_\mathrm{h}\bar v\sin\bar\psi\neq0$ and $\bar v\cos\bar\psi\neq0$, $L_Gh(x)\neq0$. The sampling box is $\bar p_x\in[-7,7]$~m, $\bar p_y\in[-5,5]$~m, $\bar\psi\in[-0.35,0.35]$~rad, and $\bar v\in[4,12]$~m/s. We use $\alpha(s)=s$, the nominal input $u_\mathrm{n}\equiv[0.5\ \ \tan0.1]^\rmT$ at the timing states, and $\SU=[-3,3]~\mathrm{m/s^2}\times[-\tan0.5,\tan0.5]$. This experiment has $m=2$ inputs. The network receives all four states, although the target barrier does not depend on $\bar p_x$. The benchmark thus has the input dimension $n=4$ of the model, and the closed-loop simulation uses a barrier that does not depend on $\bar p_x$ (Fig.~\ref{fig:closed_loop}).
		
		\textit{Arithmetic complexity.} Since $c_{\sigma_1}=c_{\sigma_2}=0$ for the ReLU, Proposition~\ref{prop:dual_cost} yields $C_\eps=4{,}799$ operations per dual pass and $3C_\eps=14{,}397$ operations for the three seeds $x_f$, $x_{G,1}$, and $x_{G,2}$, and~\eqref{eq:Csh_def} and Remark~\ref{rem:ad_cost} yield $C_\mathrm{s}=9{,}533$ and $C_\mathrm{r}=4{,}817$ operations.
		
		\textit{Measurements.} The median assembly times of the sequential entry point, of the shared-primal entry point, and of the statically allocated reverse-mode code are $97.00$, $61.61$, and $47.89$~$\mu$s, respectively. The two dual medians are thus $2.03$ and $1.29$ times the reverse-mode median, whereas the ratios of the operation counts are $2.99$ and $1.98$. The working memory is $552$, $1{,}112$, and $536$~bytes, respectively, and these three methods make no allocation. For the sequential entry point, the workspace of Proposition~\ref{prop:memory}\ref{prop2_p2} is $w_\eps=128$ floats, which occupy $512$~bytes, and the remaining $40$~bytes of the measured working memory are the call frames. The heap-allocating baseline makes $11$ allocations per call with a peak of $816$~bytes, and its median is $143.56$~$\mu$s. Of the $10{,}000$ validation states, $14$ are within $10^{-4}$ of an activation boundary, and the maximum scaled error of all four methods at these states is $2.1\times10^{-6}$.\exend
	\end{experiment}
	
	\begin{experiment}[Elliptical barrier for a Van der Pol oscillator]\label{ex:neural}
		Consider the controlled Van der Pol oscillator with state $x\in\BBR^2$, input $u\in\BBR$, and dynamics~\eqref{eq:dynamics} with
		\begin{equation}
			f(x) = \matl x_{(2)} \\ -x_{(1)}+\bigl(1-x_{(1)}^2\bigr)x_{(2)} \matr,\quad G(x) = \matl 0 \\ 1 \matr.
			\label{eq:vdp_dynamics}
		\end{equation}
		We parameterize $h_\theta$ by a ReLU network with architecture $2$-$64$-$64$-$1$, and thus $d=3$ and $C_\mathrm{f}=8{,}704$ operations. The target barrier is $1-x_{(1)}^2/9-x_{(2)}^2/4$, whose zero level set is an ellipse. The sampling box is $x\in[-4,4]\times[-3.5,3.5]$. We use $\alpha(s)=s$, $u_\mathrm{n}\equiv0$, and $\SU=[-5,5]$.
		
		\textit{Arithmetic complexity.} Proposition~\ref{prop:dual_cost} yields $C_\eps=17{,}279$ operations per dual pass and $2C_\eps=34{,}558$ operations for the seeds $x_f$ and $x_{G,1}$, and~\eqref{eq:Csh_def} and Remark~\ref{rem:ad_cost} yield $C_\mathrm{s}=25{,}854$ and $C_\mathrm{r}=17{,}284$ operations.
		
		\textit{Measurements.} The median assembly times of the sequential entry point, of the shared-primal entry point, and of the statically allocated reverse-mode code are $237.29$, $181.41$, and $169.73$~$\mu$s, respectively. The two dual medians are thus $1.40$ and $1.07$ times the reverse-mode median, whereas the ratios of the operation counts are $2.00$ and $1.50$. The working memory is $1{,}064$, $1{,}592$, and $1{,}176$~bytes, respectively, and these three methods make no allocation. Of the $10{,}000$ validation states, $78$ are within $10^{-4}$ of an activation boundary. At one of these states, the scaled errors of $L_fh_\theta$ and $L_Gh_\theta$ are $2.8\times10^{-2}$ and $2.9\times10^{-2}$ for all four methods, whereas the scaled error of $h_\theta$ is below $10^{-6}$. This state is consistent with a preactivation whose computed sign differs between single and double precision, in which case all four methods return the derivative of the adjacent affine piece (Remark~\ref{rem:kinks}). Over all $10{,}000$ states, the maximum scaled error of the four methods is thus $2.9\times10^{-2}$.\exend
	\end{experiment}
	
	\textit{Depth sweep.} To measure the dependence on the depth, we train ReLU networks of width $32$ with $2$, $4$, and $8$ hidden layers with the target, the sampling box, $\alpha$, $u_\mathrm{n}$, and $\SU$ of Experiment~\ref{ex:neural}, and thus $m=1$. Table~\ref{tab:depth} reports the median assembly time and the working memory. The working memory of the sequential and shared-primal entry points is $552$ and $824$~bytes at all three depths, whereas that of the statically allocated reverse-mode code is $536$, $824$, and $1{,}368$~bytes. The medians of the sequential entry point are $1.37$, $1.32$, and $1.31$ times the reverse-mode medians, and the medians of the shared-primal entry point are $0.98$, $0.94$, and $0.93$ times the reverse-mode medians. The ratios of the operation counts are $2.00$ for the sequential entry point and between $1.49$ and $1.50$ for the shared-primal entry point at all three depths. In the five ReLU networks of this section, the measured ratios are thus smaller than the ratios of the operation counts, and the operation count does not determine the ordering of the execution times of the shared-primal and reverse-mode implementations. The emitted loops differ in their memory accesses. In the dual passes, each loaded weight multiplies two columns in the sequential entry point and $m+2$ columns in the shared-primal entry point, and the sums accumulate in scalar variables. In the reverse-mode code, each loaded weight multiplies one vector, and the backward sweep accumulates into an array. The measurements do not separate the contributions of these memory accesses and of the arithmetic to the execution time.
	
	\begin{table}[t]
		\centering
		\caption{Depth sweep on the ESP32-S3 for ReLU networks of width $32$ with $d-1\in\{2,4,8\}$ hidden layers and $m=1$. Med.\ is the median assembly time in $\mu$s. Static, stack, and the total, which is the working memory, are in bytes.}
		\label{tab:depth}
		\footnotesize
		\setlength{\tabcolsep}{4pt}
		\begin{tabular}{llrrrr}
			\toprule
			$d-1$ & Method & Med. & Static & Stack & Total \\
			\midrule
			2 & Dual, sequential & 62.88 & 256 & 296 & 552 \\
			& Dual, shared primal & 44.96 & 384 & 440 & 824 \\
			& Reverse, static & 45.74 & 512 & 24 & 536 \\
			4 & Dual, sequential & 170.25 & 256 & 296 & 552 \\
			& Dual, shared primal & 121.49 & 384 & 440 & 824 \\
			& Reverse, static & 128.84 & 768 & 56 & 824 \\
			8 & Dual, sequential & 385.63 & 256 & 296 & 552 \\
			& Dual, shared primal & 274.67 & 384 & 440 & 824 \\
			& Reverse, static & 295.02 & 1{,}280 & 88 & 1{,}368 \\
			\bottomrule
		\end{tabular}
	\end{table}
	
	\begin{experiment}[High-order barrier for an inverted pendulum]\label{ex:pendulum}
		High-order CBFs enforce safety on systems in which the control input does not appear in $\dot h_\theta$~\citep{nguyen2016exponential, xiao2021high_order_cbf}. Consider an inverted pendulum with mass $\bar m_\mathrm{p}$, length $\bar\ell_\mathrm{p}$, state $x=[\phi\ \ \dot\phi]^\rmT\in\BBR^2$, and input torque $u\in\BBR$, where $\phi$ is the angle from the upright position. With the gravitational acceleration $g$, the dynamics~\eqref{eq:dynamics} are given by
		\begin{equation}
			f(x) = \matl \dot\phi \\ \tfrac{g}{\bar\ell_\mathrm{p}}\sin\phi \matr,\quad G(x) = \matl 0 \\ \tfrac{1}{\bar m_\mathrm{p}\bar\ell_\mathrm{p}^2} \matr.
			\label{eq:pendulum_dynamics}
		\end{equation}
		
		\textit{Relative degree and barrier design.} For a continuously differentiable $h_\theta$, the input Lie derivative is $L_G h_\theta(x)=(\partial h_\theta(x)/\partial\dot\phi)/(\bar m_\mathrm{p}\bar\ell_\mathrm{p}^2)$, and thus $h_\theta$ has relative degree one at all states $x$ such that $\partial h_\theta(x)/\partial\dot\phi\ne 0$. The network is a softplus network of architecture $2$-$32$-$32$-$1$ whose column $\col_2(W_1)$, which multiplies $\dot\phi$, is fixed at zero during training. Thus, $h_\theta$ depends on $\phi$ alone, and $L_G h_\theta\equiv 0$. The relative degree of $h_\theta$ is therefore two at all states $x$ such that $L_GL_fh_\theta(x)\ne0$, and $L_GL_fh_\theta$ vanishes at the critical points of $h_\theta$. The target barrier is $(\pi/4)^2-\phi^2$, and $\bar m_\mathrm{p}=1$~kg, $\bar\ell_\mathrm{p}=1$~m, and $g=9.81~\mathrm{m/s^2}$. Let $\alpha_1$ and $\alpha_2$ be continuously differentiable extended class $\SK$ functions, and define $\bar h_\theta(x)\triangleq L_fh_\theta(x)+\alpha_1(h_\theta(x))$. Since $L_Gh_\theta\equiv0$, the high-order CBF constraint $L_f\bar h_\theta(x)+L_G\bar h_\theta(x)u+\alpha_2(\bar h_\theta(x))\ge0$ reads
		\begin{align}
			&L_f^2 h_\theta(x)+L_G L_f h_\theta(x)\,u+\alpha_1'\bigl(h_\theta(x)\bigr)L_fh_\theta(x)\nn\\*
			&\qquad+\alpha_2\Bigl(L_f h_\theta(x)+\alpha_1\bigl(h_\theta(x)\bigr)\Bigr)\geq 0,
			\label{eq:hocbf}
		\end{align}
		which is affine in $u$. Since Theorem~\ref{thm:hyper_lie} implies that $h_\theta(x)$ and $L_fh_\theta(x)$ are the real and $\eps_1$ components of the output of each hyper-dual pass, these passes return all quantities in~\eqref{eq:hocbf}. We use $\alpha_1(s)=3s$, $\alpha_2(s)=6s$, $u_\mathrm{n}\equiv0$, and $\SU=[-20,20]$~N\,m. The sampling box is $\phi\in[-1.5,1.5]$~rad and $\dot\phi\in[-3,3]$~rad/s, and the network is trained with $\phi\in[-\pi/2,\pi/2]$~rad instead.
		
		\textit{Arithmetic complexity.} The network $\tilde h_\theta$ of Theorem~\ref{thm:hyper_lie} is evaluated at the hyper-dual seeds $x_{f,f}$ and $x_{f,G,1}$, which yields $L_f^2 h_\theta(x)$ and $L_G L_f h_\theta(x)$. Since the operation counts of Section~\ref{sec:algorithm} cover the first-order passes, this comparison counts matrix--vector products. Both seeds use $f'(x)$, whose rows are $[0\ \ 1]$ and $[(g/\bar\ell_\mathrm{p})\cos\phi\ \ 0]$. Lemmas~\ref{lem:hyper_affine} and~\ref{lem:hyper_activation} imply that each hyper-dual affine layer performs four matrix--vector products with the same weight matrix, which amounts to $4(m+1)$ products per layer for the $m+1$ seeds, and that each activation layer evaluates the scalar activation and its first two derivatives per neuron. The forward-over-reverse implementation performs two matrix--vector products per layer in each of its two sweeps, and thus four products per layer, and it returns all quantities of~\eqref{eq:hocbf}.
		
		\textit{Measurements.} The median assembly times of the hyper-dual header, of the three-component algebra, and of the statically allocated forward-over-reverse code are $510.78$, $739.40$, and $272.13$~$\mu$s, respectively, and the working memory is $1{,}208$, $968$, and $1{,}720$~bytes, respectively. The hyper-dual header thus takes $1.88$ times the execution time and $0.70$ times the working memory of the forward-over-reverse code. The three-component algebra takes $1.45$ times the execution time and $0.80$ times the working memory of the hyper-dual header. Both forms evaluate the softplus and its first two derivatives at all hidden neurons in each pass, and the three-component algebra executes three passes, whereas the hyper-dual header executes two. These three methods make no allocation, and their maximum scaled errors are $6.0\times10^{-7}$, $1.2\times10^{-6}$, and $6.0\times10^{-7}$, respectively. The finite-difference baseline has a maximum scaled error of $6.5\times10^{-4}$, which is consistent with the first-order truncation error of a one-sided difference with step $10^{-3}$, and makes $30$ allocations per call.
		
		\textit{Numerical check of the barrier.} For all $x\in\SD$, let $\rho(x)$ denote the maximum over $u\in\SU$ of the left-hand side of~\eqref{eq:hocbf}. Let $\mathcal{X}_1$ and $\mathcal{X}_2$ denote the uniform grids of $201\times201$ points in $[-1,1]\times[-1.5,1.5]$ and in $[-1,1]\times[-3,3]$, respectively, and, for all $j\in\{1,2\}$, let $\mathcal{S}_j\triangleq\{x\in\mathcal{X}_j\colon h_\theta(x)\ge0\text{ and }\bar h_\theta(x)\ge0\}$. The set $\mathcal{S}_1$ has $27{,}486$ points, and $\min_{x\in\mathcal{S}_1}\rho(x)=5.196$. The grid $\mathcal{X}_2$ covers the velocity range of the sampling box, $\mathcal{S}_2$ has $25{,}036$ points, and $\rho(x)<0$ at $1{,}797$ of them, at which~\eqref{eq:hocbf} cannot be satisfied within $\SU$. These points satisfy $|\phi|\le0.36$~rad and $|\dot\phi|\ge1.98$~rad/s. Since $h_\theta$ depends on $\phi$ alone, the input coefficient of~\eqref{eq:hocbf} is $L_GL_fh_\theta(x)=(\partial h_\theta(x)/\partial\phi)/(\bar m_\mathrm{p}\bar\ell_\mathrm{p}^2)$, which vanishes at the maximizer $\phi\approx-0.004$~rad of $h_\theta$. Furthermore, $L_f^2h_\theta(x)=(\partial^2h_\theta(x)/\partial\phi^2)\dot\phi^2+(\partial h_\theta(x)/\partial\phi)(g/\bar\ell_\mathrm{p})\sin\phi$, whose first term is negative at these points. The relative degree of $h_\theta$ is therefore not two at $\phi\approx-0.004$~rad. The target barrier has the same property at $\phi=0$. Its input coefficient $-2\phi/(\bar m_\mathrm{p}\bar\ell_\mathrm{p}^2)$ vanishes there, and, at $\phi=0$, the left-hand side of~\eqref{eq:hocbf} is $18(\pi/4)^2-2\dot\phi^2$, which is negative for $|\dot\phi|>3\pi/4\approx2.36$~rad/s. This is a numerical check at grid points and not a certificate. The experiment thus tests the second-order derivatives of the compiled network and does not establish safety on the sampling box.\exend
	\end{experiment}

	\textit{Closed-loop simulations.} Fig.~\ref{fig:closed_loop} shows sampled-data simulations of the three experiments on the host with the compiled headers. Let $t_\mathrm{s}\triangleq0.02$~s denote the control period, and, for all $k\in\BBN$, let $t_k\triangleq kt_\mathrm{s}$, and let $u_k$ denote the input that Algorithm~\ref{alg:dual_cbf} returns at the measured state $x_{\mathrm{m},k}$. The plant is integrated by a fourth-order Runge--Kutta scheme with $20$ substeps per period over $[0,12]$~s for the bicycle and the oscillator and over $[0,5]$~s for the pendulum. In the ideal case, for all $k\in\BBN$, $x_{\mathrm{m},k}=x(t_k)$, and, for all $t\in[t_k,t_{k+1})$, $u(t)=u_k$. In the perturbed case, for all $k\in\BBN$, $x_{\mathrm{m},k}=x(t_k)+e_k$, where the components of $e_k$ are independent zero-mean Gaussian samples, and, for all $t\in[t_k,t_{k+1})$, $u(t)=u_{k-1}$, where $u_{-1}\triangleq0$. The standard deviations of the components of $e_k$ are $0.02$~m for the two positions, $0.002$~rad for the heading, and $0.02$~m/s for the speed of the bicycle, $0.01$ for both states of the oscillator, and $0.002$~rad and $0.01$~rad/s for the pendulum. Each perturbed case is one realization of the noise with a fixed seed, all reported values are evaluated along the plant state, and the fallback input of Algorithm~\ref{alg:dual_cbf} does not occur. With $10$, $20$, and $40$ substeps, the states of the six filtered trajectories at the times $t_k$ differ from those with $80$ substeps by less than $2\times10^{-8}$ in the Euclidean norm, whereas forward Euler integration with the step $t_\mathrm{s}$ changes the ideal trajectory of the oscillator by up to $0.2601$. These host simulations illustrate the behavior of the filter and do not verify continuous-time or sampled-data forward invariance.
	
	Since the lane and the target barrier of the bicycle do not depend on $\bar p_x$, the filter uses the barrier $x\mapsto h_\theta(Px)$, where $P\triangleq\diag([0\ \ 1\ \ 1\ \ 1]^\rmT)$, whose Lie derivatives along $f$ and $G$ are the outputs of the compiled header at $Px$ with the dual parts $Pf(x)$ and $PG(x)=G(x)$. The network is thus evaluated at $\bar p_x=0$, inside its sampling interval of $\bar p_x$, which the bicycle leaves during the simulation. With the nominal input $u_\mathrm{n}(t)=[0.2\ \ 0.03\sin(0.5t)]^\rmT$ and $x(0)=[0\ \ 0\ \ 0\ \ 6]^\rmT$, the maximum of $|\bar p_y+\tau_\mathrm{h}\bar v\sin\bar\psi|$ is $13.73$~m without the filter and $3.429$~m and $3.427$~m with it.
	
	For the oscillator, with $x(0)=[-0.5\ \ 0]^\rmT$, the minimum of $h_\theta$ along the trajectory is $-0.8842$ without the filter, $0.1998$ in the ideal case, and $0.2176$ in the perturbed case.
	
	For the pendulum, with $u_\mathrm{n}\equiv0$ and $x(0)=[0.2\ \ 0]^\rmT$, the unfiltered trajectory leaves the band $|\phi|\le\pi/4$, and the maximum of $|\phi|$ is $0.7824$~rad and $0.7827$~rad in the two filtered cases. The trained barrier vanishes at $|\phi|\approx0.7824$~rad, and the target barrier vanishes at $|\phi|=\pi/4\approx0.7854$~rad. The minimum of $h_\theta$ along the trajectory is $6.5\times10^{-7}$ in the ideal case and $-3.5\times10^{-4}$ in the perturbed case, and the perturbed trajectory thus leaves $\{h_\theta\ge0\}$.

	\begin{figure*}[tb]
		\centering
		\includegraphics[width=.98\textwidth]{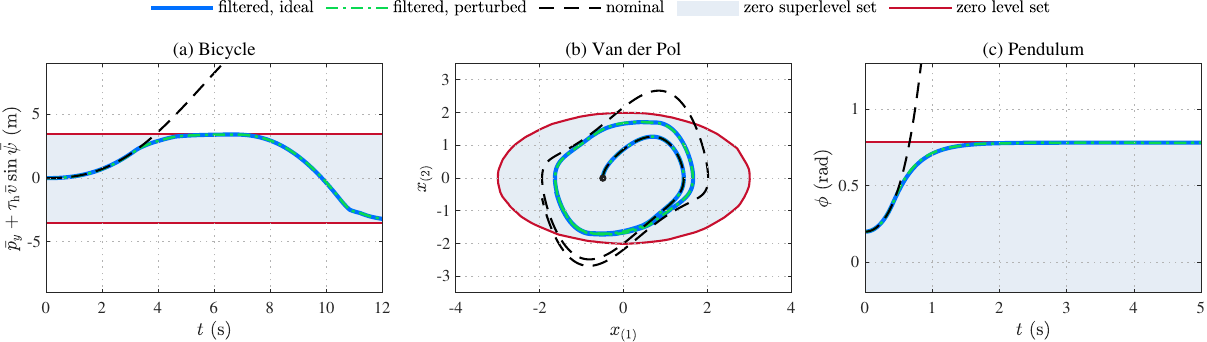}
		\caption{Sampled-data closed-loop simulations on the host with the compiled headers. (a) Look-ahead lateral position of the bicycle, (b) phase portrait of the Van der Pol oscillator, and (c) angle of the pendulum of Experiments~\ref{ex:bicycle}--\ref{ex:pendulum}. The zero superlevel set and the zero level set are those of the target barrier in (a) and (c), which is nonnegative if and only if $|\bar p_y+\tau_\mathrm{h}\bar v\sin\bar\psi|\le \bar w$ and $|\phi|\le\pi/4$, respectively, and those of the trained barrier $h_\theta$ in (b), where the zero level set consists of line segments because the ReLU network $h_\theta$ is piecewise affine, and the dot marks the initial state. The nominal trajectories are computed without the filter, and the ideal and perturbed trajectories overlap at the plotted scale.}
		\label{fig:closed_loop}
	\end{figure*}
	
	\textit{Control rates.} The step window of Table~\ref{tab:measurements} covers the filter step of Algorithm~\ref{alg:dual_cbf}, which consists of the evaluation of $f(x)$ and $G(x)$, the products with $f'(x)$ in Experiment~\ref{ex:pendulum}, the assembly of the constraint, and the solution of the QP, and it excludes sensing, communication, and actuation. The execution time of the header depends on the state through the ReLU comparisons and, for the softplus network, through the library evaluations of $\exp$ and $\log$, whose control flow is not characterized, and the execution time of the QP search depends on the state through its breakpoints and its feasibility. Neither depends on whether the state is obtained from a physical plant or from a simulation, and the sampled states of Table~\ref{tab:measurements} are drawn uniformly from the box of each experiment. The QP of the step window is the problem of the experiments, with a box and one affine constraint, and the step times do not characterize a general QP solver. The quadrotor of \citet{shayan2025ecbf} runs its position and attitude loops at $50$~Hz and $250$~Hz and its CBF-based planner at $10$~Hz. At a $250$~Hz cycle, the maximum observed step times of the sequential and hyper-dual entry points occupy $2.6\%$, $6.1\%$, and $13.0\%$ of the $4$~ms period for the bicycle, the oscillator, and the pendulum, respectively, and, at a $1$~kHz cycle, $10.5\%$, $24.4\%$, and $52.2\%$.
	
	\textit{Scope.} The experiments close the loop in simulation (Fig.~\ref{fig:closed_loop}), and no experiment with a processor in the loop, with hardware in the loop, or with a physical plant is reported. The measurements use one processor, and their use as WCET evidence requires a static analysis on the hardware of the application, for example a controller qualified under ISO~26262 \citep{debouk2019overview}. Code from the tools of Section~\ref{sec:intro} is not measured for the same network and target.
	
	\section{Conclusion}
	\label{sec:conclusion}
	
	This paper presented a compiler whose emitted C++ code evaluates a feedforward neural CBF and its Lie derivatives by forward-mode differentiation with dual numbers and without dynamic allocation. Apart from scalar variables, the workspace of a dual pass is at most four floats per neuron of the widest layer at all depths, the operation counts of the assembly of the constraint admit closed-form expressions, and hyper-dual arithmetic extends the method to the second-order Lie derivatives of CBFs of relative degree two. In the first-order comparisons, reverse mode has the lower operation count. Ahead-of-time, statically allocated reverse-mode code, which is forward-over-reverse code in Experiment~\ref{ex:pendulum}, has the lower median assembly time in the three experiments, in which the medians of the sequential and hyper-dual entry points are between $1.40$ and $2.03$ times the reverse-mode medians, and those of the shared-primal entry point are $1.29$ and $1.07$ times these medians in Experiments~\ref{ex:bicycle} and~\ref{ex:neural}, respectively, and between $0.93$ and $0.98$ times for the single-input networks of the depth sweep. The ratios from $0.93$ to $1.07$ do not establish an ordering. The working memory of the sequential entry point does not grow with the depth. It is below that of the reverse-mode code in Experiment~\ref{ex:neural} and for the networks with four and eight hidden layers, and it exceeds the latter by $16$~bytes in Experiment~\ref{ex:bicycle} and for the network with two hidden layers. No allocation is recorded while the emitted code runs, and the maximum observed execution time of the filter step with the emitted code is $522$~$\mu$s. The choice between the two modes is thus a trade-off between memory and arithmetic, which the counts of Section~\ref{sec:algorithm} quantify before a network is deployed.
	
	Since Theorem~\ref{thm:main} holds for all scalar-valued networks of the form~\eqref{eq:nn_def}, the compiler applies without change to neural control Lyapunov functions (CLFs), whose Lie derivatives CLF-CBF QPs require. Future work includes a bound on the roundoff error of the computed constraint, which a deployment requires together with a verified barrier, a static WCET analysis and processor-in-the-loop experiments on hardware for automotive and aerospace applications, and the extension of the compiler to networks with skip connections and to recurrent architectures (Example~\ref{ex:deep}).
	
	{\raggedright\printcredits\par}
	
	\section*{Declaration of competing interest}
	The author declares that he has no known competing financial interests or personal relationships that could have appeared to influence the work reported in this paper.
	
	\section*{Funding}
	This research did not receive any specific grant from funding agencies in the public, commercial, or not-for-profit sectors.
	
	\section*{Data availability}
	The compiler is available as the open-source Python package \texttt{dual-cbf-compiler} on PyPI. The source code of the compiler is available at \url{https://github.com/mkamaldar/dual_cbf_compiler}, and the experiments, which include the training scripts, the reverse-mode baselines, and the measurement tools, are available at \url{https://github.com/mkamaldar/dual_cbf_esp32_experiments}. The results were obtained with release 0.2.0 of the compiler and release 1.0.0 of the experiments repository, which are tagged at the commits \texttt{47851c0} and \texttt{e777e43}, respectively. The timing, validation, and memory states are drawn uniformly from the sampling boxes with the seeds $20261109$, $20261110$, and $20261111$, respectively, and the noise of the simulations of Fig.~\ref{fig:closed_loop} uses the seeds $1$, $2$, and $3$ for the oscillator, the bicycle, and the pendulum, respectively. The firmware prints the timing and validation seeds, the versions of ESP-IDF and GCC, and the chip revision.
	
	\appendix

	\section{Proof of Lemma~\ref{lem:invariance}}
	\label{app:invariance}

	\begin{proof}
	Let $x_0\in\SC$.
		To prove~\ref{lem1_p1}, define the closed-loop vector field $f_\mu\colon\SD\to\BBR^n$ by $f_\mu(x)\triangleq f(x)+G(x)\mu(x)$, and thus~\eqref{eq:dynamics} with $u=\mu\circ s$ is $\dot s(t)=f_\mu(s(t))$, which has no control input. Since~\ref{assum:dynamics} holds, $f$ and $G$ are locally Lipschitz continuous on $\SD$. Since, in addition, $\mu$ is locally Lipschitz continuous, and sums and products of locally Lipschitz continuous functions are locally Lipschitz continuous, $f_\mu$ is locally Lipschitz continuous on $\SD$, which implies~\ref{lem1_p1}~\citep[Thm.~2.25]{haddad2008nonlinear}.

	To prove~\ref{lem1_p2}, for all $t\in[0,t_{\max})$, define $\eta(t)\triangleq h(s(t))$. Since $h$ and $s$ are continuously differentiable, $\eta$ is continuously differentiable, and the chain rule implies that, for all $t\in[0,t_{\max})$, $\dot\eta(t)=\dot h(x,\mu(x))$ with $x=s(t)$. Since, for all $t\in[0,t_{\max})$, $s(t)\in\SD$, the hypothesis of the lemma implies that, for all $t\in[0,t_{\max})$, $\dot\eta(t)\ge-\alpha(\eta(t))$.

	For contradiction, let $t_1\in(0,t_{\max})$ be such that $\eta(t_1)<0$, and define $t_0\triangleq\sup\{t\in[0,t_1]\colon\eta(t)\ge0\}$. Since $\eta$ is continuous, $\eta(0)=h(x_0)\ge0$, and $\eta(t_1)<0$, it follows that $t_0<t_1$, $\eta(t_0)=0$, and, for all $t\in(t_0,t_1]$, $\eta(t)<0$. The mean value theorem implies that there exists $\bar t\in(t_0,t_1)$ such that $\eta(t_1)=\eta(t_0)+\dot\eta(\bar t)(t_1-t_0)$. Since $\eta(\bar t)<0$ and $\alpha$ is strictly increasing with $\alpha(0)=0$, it follows that $-\alpha(\eta(\bar t))>0$. Since, in addition, $\eta(t_0)=0$, it follows that $\eta(t_1)=\dot\eta(\bar t)(t_1-t_0)\ge-\alpha(\eta(\bar t))(t_1-t_0)>0$, which contradicts $\eta(t_1)<0$. Therefore, for all $t\in[0,t_{\max})$, $\eta(t)\ge0$, that is, $s(t)\in\SC$, which confirms~\ref{lem1_p2}.
	\end{proof}

	\section{Proofs of Proposition~\ref{thm:exact_diff} and Corollary~\ref{cor:vec_ext}}
	\label{app:thm_exact_diff}

	\begin{proofof}{Proposition~\ref{thm:exact_diff}}
	To prove~\ref{ted_p1}, observe that Definition~\ref{def:dual_ext}, the sum rule $(\varphi+\psi)'(a)=\varphi'(a)+\psi'(a)$, and the dual addition~\eqref{eq:dual_add} imply
	\begin{align}
	\widetilde{(\varphi+\psi)}(z) &= (\varphi+\psi)(a)+(\varphi+\psi)'(a)b\eps \nn\\*
	&= \bigl(\varphi(a)+\varphi'(a)b\eps\bigr)+\bigl(\psi(a)+\psi'(a)b\eps\bigr)\nn\\*
	&= \tilde\varphi(z)+\tilde\psi(z).
	\label{eq:sum_rule_dual}
	\end{align}

	To prove~\ref{ted_p2}, observe that the dual multiplication~\eqref{eq:dual_mul}, the Leibniz rule $(\varphi\psi)'(a)=\varphi'(a)\psi(a)+\varphi(a)\psi'(a)$, and Definition~\ref{def:dual_ext} imply
	\begin{align}
	\tilde\varphi(z)\,\tilde\psi(z) &= \bigl(\varphi(a)+\varphi'(a)b\eps\bigr)\bigl(\psi(a)+\psi'(a)b\eps\bigr) \nn\\*
	&= \varphi(a)\psi(a)+\bigl(\varphi'(a)\psi(a)+\varphi(a)\psi'(a)\bigr)b\eps \nn\\*
	&= \widetilde{(\varphi\psi)}(z).
	\label{eq:leibniz_dual}
	\end{align}

	To prove~\ref{ted_p3}, observe that Definition~\ref{def:dual_ext}, applied to $\varphi$ at $a$, to $\chi$ at $\varphi(a)$, and to $\chi\circ\varphi$ at $a$, and the chain rule $(\chi\circ\varphi)'(a)=\chi'(\varphi(a))\varphi'(a)$ imply
	\begin{align}
	\tilde\chi\bigl(\tilde\varphi(z)\bigr) &= \tilde\chi\bigl(\varphi(a)+\varphi'(a)b\eps\bigr) \nn\\*
	&= \chi\bigl(\varphi(a)\bigr)+\chi'\bigl(\varphi(a)\bigr)\varphi'(a)b\eps \nn\\*
	&= \widetilde{(\chi\circ\varphi)}(z).
	\label{eq:chain_dual}
	\end{align}

	To prove~\ref{ted_p4}, for all $i\in\BBN$, define $\pi_i\colon I\to\BBR$ by $\pi_i(t)\triangleq t^i$. Definition~\ref{def:dual_ext} implies that, for all $\gamma\in\BBR$, the dual extension of the constant function $\gamma$ at $z$ is $\gamma$, and thus $\tilde\pi_0(z)=1=z^0$, and that $\tilde\pi_1(z)=a+b\eps=z$. For induction, let $\kappa\in\BBZ_+$ be such that $\tilde\pi_\kappa(z)=z^\kappa$, which holds for $\kappa=1$. Since $\pi_{\kappa+1}=\pi_\kappa\pi_1$, \ref{ted_p2} implies $\tilde\pi_{\kappa+1}(z)=\tilde\pi_\kappa(z)\tilde\pi_1(z)=z^{\kappa+1}$, which confirms the claim for $\kappa+1$. Thus, for all $i\in\{0,\ldots,r\}$, \ref{ted_p2}, applied to the constant function $\gamma_i$ and to $\pi_i$, implies that the dual extension of $\gamma_i\pi_i$ at $z$ is $\gamma_iz^i$, and \ref{ted_p1}, applied $r$ times, confirms~\ref{ted_p4}.
	\end{proofof}

	\begin{proofof}{Corollary~\ref{cor:vec_ext}}
	The chain rule implies that $\psi\circ\varphi$ is differentiable at $x$ and that $(\psi\circ\varphi)'(x)=\psi'(\varphi(x))\varphi'(x)$. Let $v\in\BBR^n$, and let $z\triangleq x+v\eps$. Definition~\ref{def:dual_vec}, applied to $\varphi$ at $x$, to $\psi$ at $\varphi(x)$, and to $\psi\circ\varphi$ at $x$, thus implies
	\begin{align}
	\tilde\psi\bigl(\tilde\varphi(z)\bigr) &= \tilde\psi\bigl(\varphi(x)+\varphi'(x)v\eps\bigr) \nn\\*
	&= \psi\bigl(\varphi(x)\bigr)+\psi'\bigl(\varphi(x)\bigr)\varphi'(x)v\eps \nn\\*
	&= \widetilde{(\psi\circ\varphi)}(z),
	\label{eq:chain_restore}
	\end{align}
	which confirms~\eqref{eq:vec_chain}.
	\end{proofof}

	\section{Proofs of Lemmas~\ref{lem:dual_affine} and~\ref{lem:dual_activation}}
	\label{app:dual_layers}

	\begin{proofof}{Lemma~\ref{lem:dual_affine}}
	Since $\ell_i$ is affine, for all $a\in\BBR^{n_{i-1}}$, $\ell_i'(a)=W_i$. Definition~\ref{def:dual_vec} thus implies~\eqref{eq:dual_affine}.
	\end{proofof}

	\begin{proofof}{Lemma~\ref{lem:dual_activation}}
	The assumption on $\bar\sigma_i$ implies that $\sigma_i$ is differentiable at $a$ with the Jacobian~\eqref{eq:sigma_prime}. Definition~\ref{def:dual_vec} thus implies~\eqref{eq:dual_activation}. For all $j\in\{1,\ldots,n_i\}$, the $j$th component of~\eqref{eq:dual_activation} is~\eqref{eq:dual_RD} with $\bar\sigma_i$, $a_{(j)}$, and $c_{(j)}$ in place of $\varphi$, $a$, and $b$, respectively.
	\end{proofof}

	\section{Proofs of Theorem~\ref{thm:correctness}, Corollary~\ref{thm:safety}, and Proposition~\ref{prop:dual_cost}}
	\label{app:prop_dual_cost}

	\begin{proofof}{Theorem~\ref{thm:correctness}}
	It follows from Theorem~\ref{thm:main} that $h=\Real(\tilde h_\theta(x_f))=h_\theta(x)$ and $L_f=\Dual(\tilde h_\theta(x_f))=L_f h_\theta(x)$. It follows from~\eqref{eq:batched_eval} that $L_G=\Dual(\tilde h_\theta(X_G))=L_G h_\theta(x)$. Substituting these quantities into the inequality $L_f+L_G u\geq-\alpha(h)$ at line~6 yields the constraint of~\eqref{eq:qp}.
	\end{proofof}

	\begin{proofof}{Corollary~\ref{thm:safety}}
	Let $x\in\SD$. Since~\ref{cor2_c2} holds and, for all $i\in\{1,\ldots,d\}$, $\ell_i$ is affine, $\SD_\theta=\SD$, and the chain rule implies that $h_\theta$ is continuously differentiable on $\SD$. Theorem~\ref{thm:correctness} thus implies that Algorithm~\ref{alg:dual_cbf} returns $h_\theta(x)$, $L_f h_\theta(x)$, and $L_G h_\theta(x)$, and therefore the constraint of line~6 is the constraint of~\eqref{eq:qp} with $h=h_\theta$.

	It follows from~\ref{cor2_c3} that the condition of line~7 holds. Algorithm~\ref{alg:dual_cbf} thus executes line~8, and, since~\ref{cor2_c1} holds, the problem solved at line~8 is the QP~\eqref{eq:qp} with $h=h_\theta$, whose feasible set is nonempty, compact, and convex, because it is the intersection of the compact convex set $\SU$ with a closed half-space. The objective of~\eqref{eq:qp} is continuous and strictly convex, and therefore~\eqref{eq:qp} has the unique minimizer $u_*(x)$.

	For the case where $u_*$ is locally Lipschitz continuous, since~\ref{assum:dynamics} holds and, for all $x\in\SD$, $u_*(x)$ satisfies the constraint of~\eqref{eq:qp}, Lemma~\ref{lem:invariance} with $h=h_\theta$ and $\mu=u_*$ implies the last statement.
	\end{proofof}

	\begin{proofof}{Proposition~\ref{prop:dual_cost}}
	Consider the dual states $z_i=\hat a_i+\hat c_i\eps$ and $y_i=a_i+c_i\eps$ of Section~\ref{sec:trace}, and let $i\in\{1,\ldots,d\}$. Lemma~\ref{lem:dual_affine} implies that the real part of $\tilde\ell_i(z_{i-1})$ is the affine layer $W_i\hat a_{i-1}+b_i$ of the forward sweep and that its dual part is $W_i\hat c_{i-1}$. The product $W_i\hat c_{i-1}$ requires $n_i n_{i-1}$ multiplications and $n_i(n_{i-1}-1)$ additions. The dual part thus requires $n_i$ fewer operations than the real part, which adds $b_i$. Let $i\in\{1,\ldots,d-1\}$. Lemma~\ref{lem:dual_activation} implies that the real part of $\tilde\sigma_i(y_i)$ is the activation layer $\sigma_i(a_i)$ of the forward sweep and that its dual part is $\sigma_i'(a_i)c_i$. The evaluation of $\bar\sigma_i'$ at the $n_i$ components of $a_i$ requires $c_{\sigma_i}n_i$ operations, and the product with $c_i$ requires $n_i$ multiplications. The dual part thus requires $c_{\sigma_i}n_i$ more operations than the real part, which requires $n_i$ operations.

	The real parts of all layers thus require $C_\mathrm{f}$ operations, and the dual parts require $C_\mathrm{f}-\sum_{i=1}^{d}n_i+\sum_{i=1}^{d-1}c_{\sigma_i}n_i$ operations. The sum of these counts is~\eqref{eq:Cdf_def}. Since Algorithm~\ref{alg:dual_cbf} evaluates the dual-extended network at $m+1$ seeds, the assembly of the constraint requires $(m+1)C_\eps$ operations.
	\end{proofof}

	\section{Proofs of Proposition~\ref{thm:hyper_exact_diff}, Corollary~\ref{cor:hyper_vec}, and Lemmas~\ref{lem:hyper_affine} and~\ref{lem:hyper_activation}}
	\label{app:thm_hyper_exact_diff}

	\begin{proofof}{Proposition~\ref{thm:hyper_exact_diff}}
	Statement~\ref{hp_p1} follows from Definition~\ref{def:hyper_ext} and the linearity of the first and second derivatives. To prove~\ref{hp_p2}, observe that the arithmetic of $\BBD_2$ implies that the real, $\eps_1$, and $\eps_2$ components of $\tilde\varphi(z)\,\tilde\psi(z)$ are $\varphi\psi$, $(\varphi'\psi+\varphi\psi')b_1$, and $(\varphi'\psi+\varphi\psi')b_2$, respectively, and that its $\eps_{12}$ component is
	\begin{equation}
	(\varphi''\psi+2\varphi'\psi'+\varphi\psi'')b_1b_2+(\varphi'\psi+\varphi\psi')b_{12},
	\label{eq:app_hyper_product}
	\end{equation}
	where $\varphi$, $\psi$, and their derivatives are evaluated at $a$. The Leibniz rule implies $(\varphi\psi)'=\varphi'\psi+\varphi\psi'$ and $(\varphi\psi)''=\varphi''\psi+2\varphi'\psi'+\varphi\psi''$, and Definition~\ref{def:hyper_ext} thus implies that these components are those of $\widetilde{(\varphi\psi)}(z)$, which confirms~\ref{hp_p2}. To prove~\ref{hp_p3}, observe that~\eqref{eq:hyper_scalar}, applied to $\chi$ at $\varphi(a)$ with the components of $\tilde\varphi(z)$, implies that the real, $\eps_1$, and $\eps_2$ components of $\tilde\chi(\tilde\varphi(z))$ are $\chi(\varphi(a))$, $\chi'(\varphi(a))\varphi'(a)b_1$, and $\chi'(\varphi(a))\varphi'(a)b_2$, respectively, and that its $\eps_{12}$ component is
	\begin{align}
	&\chi''\bigl(\varphi(a)\bigr)\varphi'(a)^2b_1b_2+\chi'\bigl(\varphi(a)\bigr)\varphi''(a)b_1b_2\nn\\*
	&\qquad+\chi'\bigl(\varphi(a)\bigr)\varphi'(a)b_{12}.
	\label{eq:app_hyper_chain}
	\end{align}
	The chain rule implies $(\chi\circ\varphi)'(a)=\chi'(\varphi(a))\varphi'(a)$ and $(\chi\circ\varphi)''(a)=\chi''(\varphi(a))\varphi'(a)^2+\chi'(\varphi(a))\varphi''(a)$, and Definition~\ref{def:hyper_ext} thus implies that these components are those of $\widetilde{(\chi\circ\varphi)}(z)$, which confirms~\ref{hp_p3}.

	To prove~\ref{hp_p4}, observe that Definition~\ref{def:hyper_ext} implies that, for all $\gamma\in\BBR$, the hyper-dual extension of the constant function $\gamma$ at $z$ is $\gamma$, and that the hyper-dual extension of the identity map at $z$ is $z$. The induction in the proof of Proposition~\ref{thm:exact_diff}\ref{ted_p4}, with~\ref{hp_p1} and~\ref{hp_p2} in place of the corresponding statements of Proposition~\ref{thm:exact_diff}, thus confirms~\ref{hp_p4}.
	\end{proofof}

	\begin{proofof}{Corollary~\ref{cor:hyper_vec}}
	Let $v_1,v_2,v_{12}\in\BBR^n$, let $z\triangleq x+v_1\eps_1+v_2\eps_2+v_{12}\eps_{12}$, and define $\hat x\triangleq\varphi(x)$, $\hat v_1\triangleq\varphi'(x)v_1$, $\hat v_2\triangleq\varphi'(x)v_2$, and $\hat v_{12}\triangleq\varphi''(x)[v_1,v_2]+\varphi'(x)v_{12}$. Definition~\ref{def:hyper_vec} implies $\tilde\varphi(z)=\hat x+\hat v_1\eps_1+\hat v_2\eps_2+\hat v_{12}\eps_{12}$, and Definition~\ref{def:hyper_vec}, applied to $\psi$ at $\tilde\varphi(z)$, yields
	\begin{align}
	\tilde\psi\bigl(\tilde\varphi(z)\bigr) &= \psi(\hat x)+\psi'(\hat x)\hat v_1\eps_1+\psi'(\hat x)\hat v_2\eps_2 \nn\\*
	&\quad+\bigl(\psi''(\hat x)[\hat v_1,\hat v_2]+\psi'(\hat x)\hat v_{12}\bigr)\eps_{12}.
	\label{eq:hyper_chain_expand}
	\end{align}
	The chain rule implies that $\psi\circ\varphi$ is twice differentiable at $x$ and that $(\psi\circ\varphi)'(x)=\psi'(\hat x)\varphi'(x)$, and thus the $\eps_1$ and $\eps_2$ components of~\eqref{eq:hyper_chain_expand} are $(\psi\circ\varphi)'(x)v_1$ and $(\psi\circ\varphi)'(x)v_2$. The second-order chain rule implies $(\psi\circ\varphi)''(x)[v_1,v_2]=\psi''(\hat x)[\hat v_1,\hat v_2]+\psi'(\hat x)\varphi''(x)[v_1,v_2]$, and thus the $\eps_{12}$ component of~\eqref{eq:hyper_chain_expand} is $(\psi\circ\varphi)''(x)[v_1,v_2]+(\psi\circ\varphi)'(x)v_{12}$. Definition~\ref{def:hyper_vec}, applied to $\psi\circ\varphi$ at $z$, thus yields~\eqref{eq:hyper_chain}.
	\end{proofof}

	\begin{proofof}{Lemma~\ref{lem:hyper_affine}}
	Since $\ell_i$ is affine, for all $a,c_1,c_2\in\BBR^{n_{i-1}}$, $\ell_i'(a)=W_i$ and $\ell_i''(a)[c_1,c_2]=0$. Definition~\ref{def:hyper_vec} thus implies~\eqref{eq:hyper_affine}.
	\end{proofof}

	\begin{proofof}{Lemma~\ref{lem:hyper_activation}}
	The assumption on $\bar\sigma_i$ implies that $\sigma_i$ is twice differentiable at $a$ with the Jacobian~\eqref{eq:sigma_prime}. For all $j\in\{1,\ldots,n_i\}$, the Hessian of the $j$th component of $\sigma_i$ at $a$ has the entry $\bar\sigma_i''(a_{(j)})$ at position $(j,j)$ and zeros elsewhere, and thus $(\sigma_i''(a)[c_1,c_2])_{(j)}=\bar\sigma_i''(a_{(j)})c_{1(j)}c_{2(j)}$. Definition~\ref{def:hyper_vec} thus implies~\eqref{eq:hyper_activation}. For all $j\in\{1,\ldots,n_i\}$, the $j$th component of~\eqref{eq:hyper_activation} is~\eqref{eq:hyper_scalar} with $\bar\sigma_i$, $a_{(j)}$, $c_{1(j)}$, $c_{2(j)}$, and $c_{12(j)}$ in place of $\varphi$, $a$, $b_1$, $b_2$, and $b_{12}$, respectively.
	\end{proofof}

	\bibliographystyle{cas-model2-names}
	\bibliography{ref}
	
\end{document}